\documentclass[11pt]{article}

\usepackage[english]{babel}

\usepackage{fullpage}

\usepackage{amsmath}
\usepackage{graphicx}
\usepackage[table]{xcolor}
\usepackage[colorlinks=true, allcolors=blue]{hyperref}
\usepackage{hhline}

\usepackage{multirow} 
\usepackage{array}
\usepackage{hhline}
\usepackage{caption}
\usepackage{amsthm, amssymb, bm, graphicx, hyperref, mathrsfs}
\usepackage[algo2e,ruled,vlined,linesnumbered]{algorithm2e}
\usepackage{comment}
\usepackage[most]{tcolorbox}
\usepackage{enumitem}
\usepackage{xspace}

\usepackage{booktabs}

\usepackage{cleveref}

\providecommand{\TheoremsNumberedThrough}{}
\providecommand{\EquationsNumberedThrough}{}

\TheoremsNumberedThrough 
\EquationsNumberedThrough
\usepackage{xfrac}
\usepackage{xcolor}
\usepackage{bbm}
\usepackage{booktabs}
\usepackage{tikz}
\usepackage{subcaption}
\usepackage{float}
\makeatletter
\@ifpackageloaded{algorithm2e}{}{\usepackage[algo2e,ruled,vlined,linesnumbered]{algorithm2e}}
\makeatother

\usepackage{hhline}

\usepackage{changepage}

\usepackage{thm-restate}

\usepackage{tcolorbox}

\usepackage{csquotes}


\newtheorem{fact}{Fact}
\newtheorem{theorem}{Theorem}
\newtheorem{definition}{Definition}
\newtheorem{lemma}{Lemma}
\newtheorem{proposition}{Proposition}
\newtheorem{corollary}{Corollary}
\newtheorem{remark}{Remark}

\usepackage{mathtools, bm, amsmath}

\definecolor{cornellred}{rgb}{0.7, 0.11, 0.11}
\definecolor{maroon}{rgb}{0.52, 0, 0}
\definecolor{dgreen}{rgb}{0.0, 0.5, 0.0}
\definecolor{ballblue}{rgb}{0.13, 0.67, 0.8}
\definecolor{royalbamsmathlue(web)}{rgb}{0.25, 0.41, 0.88}
\definecolor{bleudefrance}{rgb}{0.19, 0.55, 0.91}
\definecolor{royalazure}{rgb}{0.0, 0.22, 0.66}
\usepackage{hyperref}
\hypersetup{
	colorlinks = true,
	linkcolor=royalazure,
	urlcolor=royalazure,
	citecolor=royalazure,
	linkbordercolor = {white}
}
\usepackage{cleveref}
\usepackage{enumitem}
\usepackage{multirow}

\usepackage{pgf,tikz,pgfplots}
\pgfplotsset{compat=1.15}
\usetikzlibrary{arrows}
\usepackage{pgfplots}
\usetikzlibrary{patterns}
\usepgfplotslibrary{fillbetween}
\usetikzlibrary{intersections}

\usepackage{fix-cm}

 \usepackage{natbib}

\usetikzlibrary{arrows, decorations.markings}

\tikzstyle{vecArrow} = [thick, decoration={markings,mark=at position
	1 with {\arrow[semithick]{open triangle 60}}},
	double distance=1.4pt, shorten >= 5.5pt,
	preaction = {decorate},
postaction = {draw,line width=1.4pt, white,shorten >= 4.5pt}]
\tikzstyle{innerWhite} = [semithick, white,line width=1.4pt, shorten >= 4.5pt]

\makeatletter
\newcommand{\subalign}[1]{\vcenter{\Let@ \restore@math@cr \default@tag
    \baselineskip\fontdimen10 \scriptfont\tw@
    \advance\baselineskip\fontdimen12 \scriptfont\tw@
    \lineskip\thr@@\fontdimen8 \scriptfont\thr@@
    \lineskiplimit\lineskip
    \ialign{\hfil$\m@th\scriptstyle##$&$\m@th\scriptstyle{}##$\hfil\crcr
      #1\crcr
    }}}
\makeatother

\usepackage{makecell}
\usepackage{threeparttable}

	\providecommand{\given}{}
	\DeclarePairedDelimiterX{\set}[1]\{\}{\renewcommand\given{\nonscript\:\delimsize\vert\nonscript\:\mathopen{}}#1}
	\let\Pr\relax
	\DeclarePairedDelimiterXPP{\Pr}[1]{\mathbb{P}}[]{}{\renewcommand\given{\nonscript\:\delimsize\vert\nonscript\:\mathopen{}}#1}
	\DeclarePairedDelimiterXPP{\Ex}[1]{\mathbb{E}}[]{}{\renewcommand\given{\nonscript\:\delimsize\vert\nonscript\:\mathopen{}}#1}
	
	\DeclarePairedDelimiter\ceil{\lceil}{\rceil}
	\DeclarePairedDelimiter\floor{\lfloor}{\rfloor}

\newcolumntype{P}[1]{>{\centering\arraybackslash}c{#1}}

\newcommand{\xhdr}[1]{\vspace{2mm} \noindent{\bf #1}~}

\newcommand*{\rom}[1]{\expandafter\romannumeral #1}
\newcommand{\Rom}[1]{\uppercase\expandafter{\romannumeral #1\relax}}

\newcommand{\naturals}{\mathbb{N}}

\newcommand{\OPT}{\textup{\textsc{Opt}}}
\newcommand{\ALG}{\textup{\textsc{Alg}}}
   \newcommand{\ALGBlackBoxMeta}{\textup{\textsc{BlackBoxMeta}}\xspace}
\newcommand{\ALGRoute}{\textup{\textsc{Route}}\xspace}
\newcommand{\ALGRouteOnline}{\textup{\textsc{RouteOnline}}\xspace}
\newcommand{\ALGRect}{\textup{\textsc{Rect}}\xspace}
\newcommand{\ALGRectMix}{\textup{\textsc{RectMix}}\xspace}
\newcommand{\ALGResp}{\textup{\textsc{Response}}\xspace}

\newcommand{\Greedy}{\textup{\textsc{Greedy}}\xspace} \newcommand{\ClairvoyantMix}{\textup{\textsc{Clairvoyant-Mix}}\xspace}
\newcommand{\PackAndSchedule}{\textup{\textsc{PackAndSchedule}}\xspace}

\newcommand{\PhaseDuration}{D}
\newcommand{\ClassAreaLB}{A_{\mathrm{LB}}}
\newcommand{\NumJobsClass}{n}
\newcommand{\NumJobsGe}[1]{N_{\ge #1}}
\newcommand{\NumJobsGt}[1]{N_{> #1}}
\newcommand{\LastRespPhase}{L}

\newcommand{\initialLen}{s}
\newcommand{\responseLen}{o}

\newcommand{\processingTime}{p}

\newcommand{\JobSet}{\mathcal{J}}        \newcommand{\SubJobSet}{\JobSet'}        \newcommand{\LargeJobs}{\mathcal{L}}
\newcommand{\PromptPool}{\mathcal{P}}
\newcommand{\ResponsePool}{\mathcal{R}}
\newcommand{\SmallJobs}{\mathcal{S}}
\newcommand{\CappedPromptPool}{\widehat{\PromptPool}}
\newcommand{\CappedSmallJobs}{\widehat{\SmallJobs}}
\newcommand{\EligibleJobs}{\mathcal{E}}
\newcommand{\UnfinishedJobs}{\mathcal{U}}
\newcommand{\AlgCompletedJobs}{\mathcal{C}}
\newcommand{\OptCompletedJobs}{\AlgCompletedJobs^\star}
\newcommand{\PendingQueue}{\mathcal{Q}}
\newcommand{\TruePromptPool}{\PromptPool^*}
\newcommand{\TrueResponsePool}{\ResponsePool^*}
\newcommand{\PromptHeavy}{prompt-heavy\xspace}
\newcommand{\PromptHeavyCap}{Prompt-heavy\xspace}
\newcommand{\ResponseHeavy}{response-heavy\xspace}
\newcommand{\ResponseHeavyCap}{Response-heavy\xspace}
\newcommand{\LargeBranch}{large branch\xspace}
\newcommand{\LargeBranchCap}{Large branch\xspace}
\newcommand{\PromptBranch}{prompt-heavy branch\xspace}
\newcommand{\PromptBranchCap}{Prompt-heavy branch\xspace}
\newcommand{\ResponseBranch}{response-heavy branch\xspace}
\newcommand{\ResponseBranchCap}{Response-heavy branch\xspace}

\newcommand{\paral}{k}                   \newcommand{\sliceLen}{\tau}             \newcommand{\lowerSlice}{\underline{\sliceLen}}
\newcommand{\ProxyPrompt}{\bar{\initialLen}}

\newcommand{\Index}{r}
\newcommand{\OtherIndex}{q}

\newcommand{\kvmem}{M}
\newcommand{\rectWidth}{w}
\newcommand{\Budget}{B}
\newcommand{\BudgetOf}[2]{\Budget_{#1,#2}}
\newcommand{\WidthUpperBound}{\Delta}
\newcommand{\WidthLowerBound}{\delta}
\newcommand{\PromptCutoff}{\Lambda}
\newcommand{\FullWidth}{\ell}
\newcommand{\FullWidthOf}[1]{\FullWidth_{#1}}
\newcommand{\WidthThreshold}{\lambda}

\newcommand{\StartTime}{S}
\newcommand{\WaitingTime}{W}

\newcommand{\StartTimeOf}[1]{\StartTime_{#1}}

\newcommand{\MemPeak}{\textup{\textsf{Peak}}}

\newcommand{\MSpan}[1]{{\textup{\textsc{Makespan}}}\!\left[#1\right]}
\newcommand{\C}{C}
\newcommand{\arrivalTime}{a}
\newcommand{\OptOf}[1]{\OPT(#1)}
\newcommand{\OptRectOf}[1]{\OPT_{\mathrm{rect}}(#1)}
\newcommand{\OptMSpanOf}[1]{\OPT_{\mathrm{span}}(#1)}
\newcommand{\OptArrivalOf}[1]{\OPT_{\mathrm{arr}}(#1)}
\newcommand{\Clarge}{\C^\textup{\textsc{L}}}
\newcommand{\Cprompt}{\C^\textup{\textsc{P}}}
\newcommand{\Cresponse}{\C^\textup{\textsc{R}}}
\newcommand{\rprompt}{\Index^\textup{\textsc{P}}}
\newcommand{\rresponse}{\Index^\textup{\textsc{R}}}

\newcommand{\SPS}{\textup{\textsc{SPS}}\xspace}

\newcommand{\remainJobs}{\mathcal{T}}
\newcommand{\responseScalar}{\alpha}

\newcommand{\approxratio}{\Gamma}

\newcommand{\ActiveBatch}{\mathcal{B}}
\newcommand{\curProgress}{u}

\newcommand{\Area}{A}
\newcommand{\AreaOf}[1]{\Area_{#1}}
\newcommand{\RectAreaOf}[1]{A_{#1}^{\textup{\textsc{Rect}}}}

\newcommand{\NumberJobs}{n}
 \newcommand{\setsize}[1]{{\left|#1\right|}}

\newcommand{\prob}[2][]{\text{Pr}\ifthenelse{\not\equal{}{#1}}{_{#1}}{}\!\left[{\def\givenn{\middle|}#2}\right]}
\newcommand{\expect}[2][]{\mathbb{E}\ifthenelse{\not\equal{}{#1}}{_{#1}}{}\!\left[{\def\givenn{\middle|}#2}\right]}

\newcommand{\tparen}{\big}
\newcommand{\tprob}[2][]{\text{Pr}\ifthenelse{\not\equal{}{#1}}{_{#1}}{}\tparen[{\def\given{\tparen|}#2}\tparen]}
\newcommand{\texpect}[2][]{\mathbb{E}\ifthenelse{\not\equal{}{#1}}{_{#1}}{}\tparen[{\def\given{\tparen|}#2}\tparen]}

\newcommand{\sprob}[2][]{\text{Pr}\ifthenelse{\not\equal{}{#1}}{_{#1}}{}[#2]}
\newcommand{\sexpect}[2][]{\mathbb{E}\ifthenelse{\not\equal{}{#1}}{_{#1}}{}[#2]}

\title{General Non-Clairvoyant KV-Cache Scheduling \\ via Regime-Aware Routing}

\author{
Yiding Feng\thanks{Hong Kong University of Science and Technology. Email: {\tt ydfeng@ust.hk}}
\and
Siyu Liu\thanks{Shanghai Jiao Tong University. Email: {\tt williamliusy@sjtu.edu.cn}}
\and
Zonghan Yang\thanks{Shanghai Jiao Tong University. Email: {\tt fstqwq@sjtu.edu.cn}}
\and
Yuhao Zhang\thanks{Shanghai Jiao Tong University. Email: {\tt zhang\_yuhao@sjtu.edu.cn}}
}

\date{}
\begin{document}
\maketitle

\begin{abstract}
We study non-clairvoyant scheduling for batched Large Language Model (LLM) inference under a hard Key-Value (KV) cache memory budget. Each request has a known prompt length but an unknown response length, and its memory footprint comprises a fixed prompt component together with a response component that grows with each decoded token. At each decoding round, the scheduler chooses a feasible batch of active requests; evicting a request discards its accumulated cache states, wasting prior computation. The goal is to minimize total completion time against the optimal clairvoyant schedule that knows all response lengths.

We present the first $O(1)$-competitive algorithm for arbitrary prompt lengths and arbitrary response lengths with no additional assumptions. Rather than relying on a single universal scheduling policy, our algorithm is built on a novel \emph{regime-aware routing} framework. Specialized sub-schedulers handle different memory-growth geometries, while a meta-scheduler time-shares the memory budget across them and dynamically routes each job as its execution progressively reveals its behavior. This framework also yields constant-factor guarantees for makespan and for total completion time under online arrivals.
\end{abstract}

\section{Introduction}
\label{sec:introduction}

{
Large Language Model (LLM) inference has become a cornerstone of modern computing infrastructure, supporting interactive assistants, automated workflows, and agentic applications at massive scale.
As these workloads shift from short exchanges toward longer and more computation-intensive tasks, the inference phase has become a major contributor to operational cost and energy consumption~\citep{LJS24,SZGCQFT25,SZGTC25}.
Meeting quality-of-service requirements under these resource pressures calls for careful scheduling across the serving stack, from cluster-wide request routing to the memory management logic on individual serving instances.

This work focuses on replica-level scheduling inside an LLM inference service.
In modern LLM serving stacks, a routing layer assigns each incoming request to one of many parallel serving instances, known as \emph{replicas};
within each replica, the local scheduler decides which requests to process together and how to manage the finite memory budget~\citep{YJKKC22,KLZSZYGZS23,ZYXS0YCKSGB24,SHZXZL024}.
The central challenge arises from the Transformer architecture, whose self-attention mechanism stores intermediate states in a per-request \emph{Key-Value (KV) cache} to avoid quadratic recomputation during decoding~\citep{VSPUJGKP17,PDCDBHXAD23}.
Processing a request involves a \emph{prefill} phase that encodes the input prompt and initializes the cache, followed by a \emph{decode} phase that generates response tokens autoregressively, appending new KV entries at each step.
The memory footprint of a request therefore has a fixed component from the prompt and a component that grows linearly with the response length.
Modern systems manage this dynamic growth through two mechanisms: \emph{continuous batching}~\citep{YJKKC22}, which updates the active batch at the granularity of individual decoding iterations, and \emph{block-based cache management}~\citep{KLZSZYGZS23,ZYXS0YCKSGB24}, which allocates and recycles memory across requests.
When the budget is exhausted, active requests must be \emph{killed}, discarding the KV-cache entries accumulated during decoding.
The scheduler must therefore balance high concurrency, which improves throughput, against the cost of kills that waste prior computation.
}

This tension is compounded by a second critical difficulty: the scheduler is \emph{non-clairvoyant}.
Response lengths are not known in advance but are determined by the model's generation process and revealed only at completion.
The autoregressive and stochastic nature of text generation makes exact prediction inherently impossible, even with dedicated predictors~\citep{ZRXL0Y23,FZSQS024,QMPCJWFKBI24}.
In classical non-clairvoyant scheduling, a job of unknown size occupies a fixed amount of resource whenever it runs; here, uncertainty in response length directly translates into uncertainty in future memory demand, because the cache grows with every decoded token.
The scheduler must therefore manage dynamically growing resource requirements under incomplete information---a doubly constrained trade-off between latency and memory safety.

This places the problem in the non-clairvoyant setting of scheduling theory~\citep{MPT94,BL04}, with the additional feature that the resource demand of a job increases during its execution.
For this problem, a recent line of theoretical work~\citep{JJMMPZ25,WYZ25,CYZ25,FYZ26,KQYZ-26} has established guarantees for the single-replica model under restricted regimes, such as clairvoyant scheduling, identical prompt lengths, or large-memory settings.
Removing these restrictions, we ask:

\begin{quote}
\emph{Can a non-clairvoyant scheduler perform comparably to the optimal schedule computed with full knowledge of all response lengths?}
\end{quote}

To make this question concrete, we adopt the standard offline batch model studied in recent theoretical work on LLM inference scheduling~\citep{WYZ25,CYZ25,FYZ26,KQYZ-26}.
A single serving instance, or replica, has a hard memory budget~$\kvmem$.
A set of jobs (inference requests) arrives at time zero.
Each job has a known prompt length $\initialLen_i$ and an unknown response length $\responseLen_i$.
Time proceeds in discrete rounds.
In each round, the scheduler chooses a batch of unfinished jobs to process; every job in the batch decodes one token.
The memory consumed by a job equals its prompt length plus the number of tokens it has already decoded in its current uninterrupted attempt.
The total memory of all active jobs must not exceed the budget $\kvmem$.
The scheduler may also kill an active job, discarding all progress from that attempt and forcing a restart.
The objective is to minimize total completion time---the sum of the completion times of all jobs---measured against an optimal clairvoyant schedule that knows all response lengths in advance.
A non-clairvoyant algorithm is \emph{competitive} if its total completion time is at most a constant factor times this optimum on every instance.

Compared to classical non-clairvoyant scheduling, the defining feature of our model is that resource demands are dynamic rather than static.
A batch that is feasible at the start of a round may become infeasible as its members grow, forcing the scheduler to kill jobs to free memory or to run fewer jobs in parallel.

\subsection{Our Results and Techniques}
\label{sec:intro:results-and-techniques}
We answer the aforementioned question affirmatively. In a nutshell, we give the \emph{first constant-competitive non-clairvoyant} scheduling algorithm for the offline batch model with arbitrary prompt lengths and arbitrary response lengths, requiring no additional assumptions on the problem instances. Formally, our main result is as follows.

\begin{restatable}{theorem}{thmMain}
\label{thm:main}
For every feasible batch KV-cache scheduling instance $\JobSet$ with arbitrary prompt lengths and arbitrary response lengths, there exists a fully non-clairvoyant scheduling algorithm whose total completion time is within a constant factor of the optimal clairvoyant total completion time:
\begin{equation*}
    \ALG(\JobSet)
    =
    O(1)\cdot \OPT(\JobSet) .
\end{equation*}
\end{restatable}
Beyond satisfying the competitive guarantee stated in the theorem, our algorithm runs in polynomial time and is also constant-competitive for the alternative objective of makespan.\footnote{In our model, makespan is a relatively easier objective; we therefore focus primarily on completion time (i.e., latency), which is also the more practically important metric in LLM inference.}

In the remainder of this section, we give an overview of the key challenges and the main technical ingredients of our algorithm: \emph{regime decomposition}, which partitions jobs into geometric regimes; \emph{budget-doubling}, which handles unknown response lengths within each regime; and a \emph{routing meta-scheduler}, which combines the regimes with only constant-factor overhead.

\xhdr{From area order to regime decomposition.}
A natural starting point is the clairvoyant setting, where the scheduler knows all response lengths in advance.
Following prior work \citep{FYZ26}, we define the \emph{(memory-time) area} of a job as its total memory occupation over a successful run, i.e., $\AreaOf{i}\triangleq \initialLen_i\cdot\responseLen_i + \responseLen_i(\responseLen_i+1)/2$.
Below, we refer to memory-time area simply as \emph{area}.
Note that the optimal total completion time can be lower bounded by sorting all jobs by nondecreasing area and charging each job for the cumulative area that must be supplied before it can finish: the $r$-th smallest job cannot complete before the system has delivered the total area of the first $r$ jobs, which takes time proportional to the sum of their areas divided by the memory budget (see \Cref{lem:prelim:area-lower-bound}).
This bound holds because the system can supply at most $\kvmem$ units of memory-time area per round, so the aggregate area of all jobs must be paid for in time.

The area lower bound not only provides a tractable way to characterize the optimal clairvoyant benchmark; more importantly, it also suggests a sufficient condition for a constant approximation or competitiveness:
\begin{itemize}
    \item \textbf{(High memory utilization)} the schedule should keep a constant fraction of the memory occupied; and
    \item \textbf{(Area-order completion)} the completion order should be close to nondecreasing memory-time area order.
\end{itemize}
One might hope to achieve both with a single priority rule over job types---for example, greedy scheduling by area.\footnote{Specifically, put all jobs in a priority queue based on their area. At any time step, if the job at the top of the priority queue can be added to the current batch while respecting the memory constraint, remove it from the queue and add it to the processing batch.} Indeed, all known non-clairvoyant algorithms in the related restart and KV-cache scheduling models, and the majority of clairvoyant KV-cache scheduling algorithms, derive their schedules from a global ordering criterion~\citep{JSWW25,FYZ26,JJMMPZ25,CYZ25,KQYZ-26}.
However, we show that this is impossible even with full knowledge of all response lengths in \Cref{prop:single-rule-fails}.
The obstruction has two complementary forms.
If the scheduler prioritizes small-area jobs, it may leave the memory severely underutilized while a long small-prompt job monopolizes the batch.
If it prioritizes large-area jobs, it may block many small-area jobs behind a few large-prompt jobs that must run serially.
These two failure modes impose opposite pressures: high utilization and area-consistent ordering cannot be reconciled by one global priority.
This motivates a decomposition into geometric regimes that are treated separately.

We partition the jobs using a prompt cutoff $\PromptCutoff = \kvmem/4$ into three regimes:
\begin{itemize}
    \item \textbf{(Large jobs)} $\initialLen_i > \PromptCutoff$. These jobs are so wide that each one nearly fills the memory budget alone.
    \item \textbf{(Small prompt-heavy jobs)} $\initialLen_i \le \PromptCutoff$ and $\responseLen_i \le \initialLen_i$. The response is no longer than the prompt, so the memory footprint stays within a constant factor of the prompt length.
    \item \textbf{(Small response-heavy jobs)} $\initialLen_i \le \PromptCutoff$ and $\responseLen_i > \initialLen_i$. The memory footprint grows well beyond the prompt, and the area is dominated by the response length.
\end{itemize}
This three-way partition is the regime decomposition used by the algorithm:
jobs are classified by their geometric properties and sent to distinct scheduling pipelines, each specialized to a particular memory-growth profile.

Each regime, taken in isolation, admits a constant-competitive non-clairvoyant algorithm, which we describe in the next two parts.
The remaining challenge is that the non-clairvoyant scheduler does not know which regime a job belongs to, because the distinction between prompt-heavy and response-heavy depends on the unknown response length.
We explain how to resolve this difficulty when we discuss the meta-scheduler below.

\xhdr{Rectangle scheduling for large and prompt-heavy jobs.}
For large jobs and prompt-heavy jobs, the KV-cache footprint grows by at most a constant factor during execution.
Specifically, if $\responseLen_i \le \initialLen_i$, then the memory usage at every timestep lies between $\initialLen_i + 1$ and $2\initialLen_i$.
This allows us to abstract each job as a rectangle of fixed width $\rectWidth_i\triangleq 2\initialLen_i$ and unknown processing time $\processingTime_i\triangleq \responseLen_i$, where killing an active attempt discards all progress.\footnote{Here we use terminology from strip packing and resource-constrained scheduling~\citep{LMM02,JR21}.}
To our knowledge, this \emph{non-clairvoyant} version of rectangle strip scheduling has not been studied in the literature, and our results here may be of independent interest.

Our scheduler $\ALGRect$ (\Cref{alg:rectangle-strip}) maintains a priority queue of attempts ordered by their \emph{area budget} $\rectWidth_i \cdot 2^\Index$, where $2^\Index$ is the geometrically increasing length cap for the $\Index$-th attempt of job $i$.
At each event time, the scheduler greedily admits the minimum-budget attempt if it fits in the remaining memory.
If an attempt reaches its cap without completing, it is killed and replaced by a doubled attempt, which is added back to the priority queue with the updated area budget.

The analysis rests on a \emph{queue monotonicity} (\Cref{lem:rect-queue-monotonicity}): the minimum area budget among pending attempts is weakly increasing over time, and every active attempt has area budget no larger than every pending attempt.
This invariant lets us split the completion time of each job into active time and waiting time.
The active time---the total time during which some (possibly unsuccessful) attempt of job $i$ is executing, so that no attempt of that job is waiting in the queue---is bounded by $4\processingTime_i$ because the geometric caps up to and including the first successful one sum to less than $4\processingTime_i$.
The waiting time---the total time during which an attempt of job $i$ is in the priority queue---is controlled by a utilization bound: when all rectangle widths lie in a bounded range $[\WidthLowerBound, \WidthUpperBound]$, every waiting round has active memory of at least $\max\{\WidthLowerBound, \kvmem-\WidthUpperBound\}$.
For large jobs, $\WidthLowerBound = 2\PromptCutoff = \Omega(\kvmem)$, so utilization is constant, as desired; for prompt-heavy jobs, $\WidthUpperBound = 2\PromptCutoff$, and utilization is at least $\kvmem - 2\PromptCutoff = \Omega(\kvmem)$, again as desired.
Moreover, the queue monotonicity ensures that the completion order approximately follows the area order.
Consequently, the scheduler $\ALGRect$ (\Cref{alg:rectangle-strip}) is $36$-competitive when all jobs are large, and also when all jobs are small prompt-heavy.

Remarkably, even a few small response-heavy jobs suffice to break the guarantee: memory may still be well utilized, but jobs can now finish in the wrong order---with large-area jobs completing before small-area ones---because the queue monotonicity no longer holds. This ordering failure makes the competitive ratio unbounded.

\xhdr{Geometric slicing for response-heavy jobs.}
When the response length dominates the prompt, the memory footprint grows far beyond the prompt length, so the fixed-width rectangle view no longer applies.
Instead, the area of a response-heavy job is governed by its response length: for each job~$i$, the area $\AreaOf{i}\triangleq \initialLen_i\cdot\responseLen_i + \responseLen_i(\responseLen_i+1)/2$ lies between $\responseLen_i^2/2$ and $2\responseLen_i^2$.
Thus response-length order and area order agree up to a constant factor.

We handle these jobs with a scheduler $\ALGResp$ (\Cref{alg:response-heavy}) that extends the geometric-slicing approach of \citet{FYZ26} from identical to non-identical prompt lengths.
The algorithm proceeds in phases $\Index = 0, 1, \dots$ with geometrically increasing response caps $\sliceLen_{\Index} = 2^{\Index}$.
In each phase, every unfinished job with $\initialLen_i \le \sliceLen_{\Index}$ becomes \emph{eligible}.
All eligible jobs are then scheduled using the SPS (Staggered Pipeline Scheduling) subroutine of \citet{FYZ26}, which packs jobs with identical prompt and response lengths into a high-utilization batch.
We invoke SPS by treating all eligible jobs as sharing a common proxy prompt $\ProxyPrompt_{\Index} = \min\{\sliceLen_{\Index}, \kvmem-\sliceLen_{\Index}\}$ and the current cap $\sliceLen_{\Index}$ as their common response length.

The key insight is that within a phase, the shared proxy prompt lets SPS guarantee high memory utilization.
Moreover, since the caps grow geometrically, each job completes in the first phase whose cap exceeds its true response length.
This means the completion order induced by the phases approximately follows the response-length order, which in this regime is close to the area order.
Consequently, the scheduler is $236/3 \approx 78.67$-competitive when all jobs are response-heavy.

Remarkably, even a few prompt-heavy jobs are enough to break the guarantee.
If we keep the same phase-based eligibility (i.e., admitting a job whenever its prompt fits under the current response cap), the SPS schedule can leave memory underutilized because prompt-heavy jobs consume far less cache than their proxy length suggests.
In addition, response-length order and area order can diverge dramatically for such jobs: a short-response job with a long prompt may have large area, while a long-response job with a short prompt may have small area.
Either mismatch alone makes the competitive ratio of our scheduler $\ALGResp$ (\Cref{alg:response-heavy}) unbounded.

\xhdr{Meta-scheduling with routing.}
So far, we have described separate schedulers for large jobs, small prompt-heavy jobs, and small response-heavy jobs.
From an engineering perspective, if we had three replicas---each with its own memory budget of $\kvmem$---we could route each regime to a dedicated replica and immediately obtain an $O(1)$-competitive guarantee, provided the jobs were partitioned correctly.
The remaining challenge is to achieve the same guarantee on a single replica, without knowing the partition in advance.

If the three-way partition were known, we could combine the sub-schedulers via a black-box meta-scheduler.
The meta-scheduler $\ALGBlackBoxMeta$ (\Cref{alg:black-box-meta}) runs the sub-schedulers in round-robin order over doubling stages of length $2^\Index$.
Each sub-scheduler receives a fresh copy of its job set in each of its allotted rounds.
Completed jobs are treated as dummy jobs in later stages.
A per-job completion bound shows that this time-sharing loses only a constant factor relative to running each sub-scheduler in isolation.

In the non-clairvoyant setting, however, the partition is not known: for a job with small prompt, the scheduler cannot tell whether $\responseLen_i \le \initialLen_i$ or $\responseLen_i > \initialLen_i$ without running it.
We therefore use a routing meta-scheduler $\ALGRoute$ (\Cref{alg:white-box-meta}).
Small jobs are initially assigned to the prompt-heavy branch.
If a job completes before generating $\initialLen_i$ tokens, it is finished; if it reaches $\initialLen_i$ tokens without completing, it is \emph{certified} as response-heavy and routed to the response-heavy branch.
Large jobs are sent directly to the large-job branch.

The routing meta-scheduler uses the same doubling stages as the black-box version.
In each stage, it runs the prompt-heavy branch, then the response-heavy branch (which now includes any newly certified jobs), then the large branch.
The key observation is that a job certified in the prompt-heavy branch during stage $\Index$ is available to the response-heavy branch in the same stage, and the doubling lengths ensure that both certification and completion are accommodated within a constant factor of their isolated completion times (\Cref{lem:routing-loss}).
Combining the three branch guarantees yields the final $O(1)$-competitive guarantee for all instances with arbitrary prompt lengths and arbitrary response lengths stated in \Cref{thm:main}.

\xhdr{Variants and extensions.}
Beyond our main non-clairvoyant result stated in \Cref{thm:main}, we establish several complementary guarantees using our algorithmic framework to chart the broader landscape of KV-cache scheduling.
First, in the clairvoyant setting, our framework yields a $16$-approximation algorithm for all instances with arbitrary prompt lengths and arbitrary response lengths (\Cref{thm:clairvoyant-rectangle-packing}), and a $(3+o(1))$-approximation for the large-memory regime where every job satisfies $\initialLen_i + \responseLen_i = o(\kvmem)$ (\Cref{cor:clairvoyant-large-memory}).\footnote{The same approximation factor of $3$ for large-memory instances has been established in concurrent work by \citet{KQYZ-26}.}
Second, the same routing algorithm yields an $O(1)$-competitive ratio for the makespan objective (\Cref{thm:alternative-makespan}), showing that the structural benefits of regime decomposition extend beyond total completion time.
Third, we extend our results to online arrivals, where jobs arrive over time and reveal their prompt lengths upon arrival: we give an online variant of our routing meta-scheduler that remains $O(1)$-competitive for total completion time (\Cref{thm:alternative-online-completion}).\footnote{For the standard flow-time objective under adversarial arrivals, strong impossibility results hold: even restricted settings admit super-constant lower bounds on the competitive ratio \citep{JJMMPZ25}. Total completion time is therefore the natural latency objective for constant-competitive guarantees in the online non-clairvoyant setting.}

\xhdr{System insights.}
Modern LLM inference systems already allow global request routing across different replicas.
Our results provide a principled foundation for how this routing should be done: partition jobs according to their geometric structure---large jobs, small prompt-heavy jobs, and small response-heavy jobs.
This partition can be implemented at the global router level, where prompt lengths are known and response-based classifications can be updated dynamically as execution reveals job behavior.
Once jobs are routed to separate replicas, each local scheduler can simply run the corresponding sub-scheduler from our framework.
This avoids the overhead of the theoretical meta-scheduler, which must time-share a single memory budget across heterogeneous classes.
Consequently, such an implementation achieves a better competitive ratio than the single-replica theoretical bound, while remaining easy to deploy.

The geometric regime decomposition explains why the existing separation between routing and local scheduling is algorithmically useful.
Simple local rules, properly routed by geometric class, can achieve both worst-case robustness and practical efficiency.

\subsection{Organization}
\Cref{sec:preliminaries} defines the scheduling model and develops an area-based lower bound on the optimal completion time.
\Cref{sec:regimes} classifies jobs into three regimes---large, \PromptHeavy, and \ResponseHeavy---and motivates treating them separately through explicit hard instances.
\Cref{sec:rectangle-strip} gives the rectangle scheduler $\ALGRect$ (\Cref{alg:rectangle-strip}), which handles large and \PromptHeavy jobs by exploiting their nearly constant memory footprint, while \Cref{sec:response-heavy} gives the response scheduler $\ALGResp$ (\Cref{alg:response-heavy}), which handles the growing footprint of \ResponseHeavy jobs through area-geometric slicing with doubling response caps.
\Cref{sec:meta-scheduling} presents the black-box meta-scheduler $\ALGBlackBoxMeta$ (\Cref{alg:black-box-meta}) that time-shares the memory budget across independent sub-schedulers, and extends it to the routing meta-scheduler $\ALGRoute$ (\Cref{alg:white-box-meta}) that resolves the unknown job partition during execution, completing the proof of \Cref{thm:main}.
Finally, \Cref{sec:clairvoyant,sec:rectangle-strip-complete,sec:response-heavy-constant,sec:alternative-objectives} contain the clairvoyant approximation guarantees, deferred analyses for the rectangle and \ResponseHeavy schedulers, and extensions to makespan and online-arrival completion time.

\subsection{Related Work}

\xhdr{KV-cache scheduling.}
{\citet{JJMMPZ25} initiate the theoretical study of KV-cache scheduling for LLM inference by formulating the scheduling model, 
proving an $\Omega(\sqrt{\NumberJobs})$ lower bound for deterministic algorithms under adversarial online arrivals,
and giving the Memory-Constrained Shortest-First rule which achieves a $9216$-approximation in the offline identical-prompt clairvoyant setting.
In the clairvoyant setting, \citet{WYZ25} study the large-memory regime with heterogeneous prompt and response lengths and give a polynomial-time $(48+\varepsilon)$-approximation for any fixed $\varepsilon>0$.
Recently, concurrent work by \citet{KQYZ-26} improves the approximation ratio for the clairvoyant large-memory regime to $3+o(1)$.
For non-clairvoyant settings, \citet{CYZ25} study scheduling under output-length uncertainty and give an $O(\log \kvmem)$-competitive algorithm in the large-memory regime.
\citet{FYZ26} give a non-clairvoyant geometric-slicing algorithm for the identical-prompt offline batch model, with a competitive ratio of at most $61.92$ in general and $32$ in the large-memory regime.
Parallel lines of work~\citep{ALSW25,LDP25} study stochastic arrivals through fluid-guided online scheduling and queueing-throughput viewpoints with different optimization objectives.}

\xhdr{Non-clairvoyant scheduling.}
Classical non-clairvoyant scheduling is mainly studied in a preemptive pause-and-resume model, where processing times are unknown until completion and a suspended job keeps its completed work.
\citet{MPT94} introduced this model and gave deterministic and randomized upper and lower bounds; \citet{BL04} developed randomized multilevel-feedback algorithms for total flow time on single and parallel machines.
Later variants add limited side information, including approximate processing-time classes~\citep{BLMP04}, job-size predictions~\citep{IKQP23}, predictions for only a subset of jobs~\citep{BP24}, and progress-bar feedback~\citep{BCLS25}.
\citet{ImKM18} introduce the packing/polytope scheduling framework, where active jobs receive processing rates subject to packing constraints on the rate vector, and give non-clairvoyant algorithms including a constant-competitive guarantee for total weighted completion time.
\citet{JLM26} sharpen the proportional-fairness analysis for this framework, including important monotone special cases, under the same total weighted completion-time objective.

Another line of research studies restart models, where interrupting a job destroys its accumulated progress.
In clairvoyant online settings, restarts have been studied for single-machine total completion time and for completion-time and flow-time lower bounds~\citep{SP05,ES03}.
\citet{JSWW25} study single-machine scheduling with a weighted completion-time objective in a non-clairvoyant restart model, prove lower bounds for deterministic restart strategies, and give tight analyses for geometric restart strategies.

\xhdr{Completion-time scheduling.}
Total and weighted completion time are central min-sum objectives in scheduling.
For the basic single-machine weighted case, \citet{Smith56} gives an exact ordering principle by sorting jobs according to processing time per unit weight.
For richer machine environments, \citet{HSSW97} develop general offline and online approximation techniques for completion-time objectives, using linear-programming relaxations to guide list scheduling.
\citet{CMNS01} give ordering and rounding methods for average completion time with release dates, parallel machines, and precedence constraints; their algorithms derive an order from a preemptive or linear-programming relaxation and round it into a nonpreemptive schedule.
More recently, \citet{Li20} shows the strength of time-indexed linear-programming relaxations for total weighted completion time, including precedence-constrained machine scheduling; using completion-time structure as a guide, \citet{CIP25} give a gradient-descent meta-algorithm based on the residual weighted completion-time optimum and obtain scalable algorithms for weighted flow time.

\xhdr{Geometric strip packing \& scheduling.}
The memory-time view represents executions as growing geometric regions packed under a fixed memory width.
The KV-cache regions in this model are trapezoidal, while related geometric-packing literature studies rectangles and more general polygons, often with objectives such as the number of bins or the strip length.
\citet{LMM02} survey two-dimensional packing models, including bin packing and strip packing variants.
\citet{AKPP17} prove APX-hardness for strip packing, including the case with polynomially bounded input data.
\citet{Steinberg97} gives a $2$-approximation for rectangle strip packing.
\citet{JZ07} give a $(3+\epsilon)$-approximation for maximum-profit rectangle packing in a bounded rectangle.
More general polygon packing has also been studied: \citet{ABK17} give approximation algorithms for packing convex polygons into minimum-area containers, and \citet{KS23} improve approximation guarantees for translational packing of convex polygons, including polygon strip packing.

Resource-constrained scheduling gives a closely related strip view: a job with processing time and resource demand occupies a rectangle in the time-resource strip, so makespan corresponds to strip length and total completion time corresponds to the sum of rectangle endpoint times.
Under this view, \citet{JR21} give a tight $(3/2+\epsilon)$-approximation for the unweighted makespan objective on identical machines with one shared renewable resource, while \citet{CCZ24} give constant-approximation algorithms for the weighted total completion-time objective on capacitated parallel machines, using volume-by-weight ordering and a hybrid treatment of large resource demands.
 \section{Preliminaries}
\label{sec:preliminaries}

In this section, we define the scheduling model, the non-clairvoyant information structure, the total completion-time objective, and the benchmark bounds used throughout the paper.
Our notation mostly follows the offline batch setting studied in the recent line of work on LLM inference scheduling~\citep{JJMMPZ25,FYZ26}.

\xhdr{Jobs and scheduling actions.}
The system has a single serving instance with a KV-cache memory budget $\kvmem\in\naturals^{+}$.
There are $\NumberJobs$ jobs, denoted by $\JobSet\triangleq[\NumberJobs]$.
Job $i\in\JobSet$ has a known prompt length $\initialLen_i\in\naturals^{+}$ and an unknown response length $\responseLen_i\in\naturals^{+}$.
We assume that each job is individually feasible:
\begin{align*}
    \initialLen_i+\responseLen_i\le \kvmem,
    \qquad \forall i\in\JobSet .
\end{align*}
All jobs arrive at time zero.
Time is discrete and proceeds in rounds $t=0,1,2,\ldots$, where time $t$ denotes the beginning of round $t$.
At time $t$, each unfinished job is either active in a current attempt or inactive.
For each active job $i$, let $\curProgress_{i,t}$ denote the number of response tokens already decoded in its current uninterrupted attempt before round $t$ begins.
For notational convenience, define $\curProgress_{i,t}=0$ when job $i$ is inactive.
At the beginning of each round $t$, the scheduler chooses an active batch $\ActiveBatch_t\subseteq\JobSet$ of unfinished jobs.
Every job $i\in\ActiveBatch_t$ receives one unit of processing during the round.
If $\curProgress_{i,t}+1=\responseLen_i$, then job $i$ completes at the end of round $t$, and its completion time is $t+1$.
Before deciding $\ActiveBatch_t$, the scheduler may \emph{kill} any active unfinished job $i$, discarding all progress in its current attempt and setting $\curProgress_{i,t}=0$; every unfinished job in $\ActiveBatch_{t-1}\setminus\ActiveBatch_t$ is implicitly killed.

\xhdr{KV-cache memory feasibility.}
If job $i$ is active in round $t$, its KV-cache memory usage is
\begin{align*}
   \initialLen_i+\curProgress_{i,t}+1,
\end{align*}
where the term $+1$ accounts for the token being decoded in the current round.
A schedule is feasible if every active batch satisfies
\begin{align*}
    \sum\nolimits_{i\in\ActiveBatch_t}
    \left(\initialLen_i+\curProgress_{i,t}+1\right)
    \le \kvmem,
    \qquad \forall t\ge 0 .
\end{align*}
In other words, the total memory occupied by all active jobs (i.e., their fixed prompts $\sum\nolimits_{i\in\ActiveBatch_t}\initialLen_i$ plus the KV entries generated so far $\sum\nolimits_{i\in\ActiveBatch_t}\curProgress_{i,t}$ plus the tokens currently being decoded $|\ActiveBatch_t|$) must never exceed the memory budget $\kvmem$.

\xhdr{Non-clairvoyant setting and competitive analysis.}
The non-clairvoyant scheduler knows the number of jobs $\NumberJobs$, memory budget $\kvmem$, and the prompt lengths $\{\initialLen_i\}_{i\in\JobSet}$.
It does not know the response lengths $\{\responseLen_i\}_{i\in\JobSet}$ in advance.
The value of response length $\responseLen_i$ is revealed only when job $i$ completes.
The benchmark scheduler is clairvoyant and knows all response lengths in advance.

The objective is to minimize \emph{total completion time}, i.e., the sum of the completion times of all jobs.
We evaluate non-clairvoyant algorithms against the optimal clairvoyant schedule, which knows all response lengths in advance.
Let $\OPT(\JobSet)$ denote the minimum total completion time achievable by such a clairvoyant schedule\footnote{We allow the optimal clairvoyant scheduler to preempt or pause jobs (i.e., retain a job in memory without processing it). However, neither is necessary in an optimal solution: delaying a job's start time weakly dominates pausing, and preemption provides no advantage.} on the full instance~$\JobSet$.\footnote{With slight abuse of notation, we also use $\JobSet$ to denote both the set of jobs and the corresponding instance.}
For an algorithm $\ALG$, let $\ALG(\JobSet)$ denote its total completion time on~$\JobSet$ starting from time~$0$.
We say that $\ALG$ is $\approxratio$-competitive for $\approxratio\geq 1$ if, for every feasible instance,
\begin{align*}
    \ALG(\JobSet) \le \approxratio \cdot \OPT(\JobSet) .
\end{align*}

For the analysis, we extend these definitions to any subset $\SubJobSet\subseteq\JobSet$.
Let $\OptOf{\SubJobSet}$ denote the minimum total completion time achievable by a clairvoyant schedule that processes exactly the jobs in $\SubJobSet$, and let $\ALG(\SubJobSet)$ denote the total completion time of $\ALG$ on $\SubJobSet$ starting from time~$0$.

An important concept in our algorithms and analysis is the \emph{memory-time area}, which we abbreviate simply as \emph{area} when the context is clear.

\begin{definition}[Memory-time area]
For a job with prompt length $\initialLen$ and response length $\responseLen$, its \emph{memory-time area} is the cumulative KV-cache memory occupied over a successful run:
\begin{align*}
    \Area(\initialLen,\responseLen)
    \triangleq
    \initialLen\cdot \responseLen
    +\frac{\responseLen\cdot (\responseLen+1)}{2}.
\end{align*}
For job $i$, we write $\AreaOf{i} \triangleq \Area(\initialLen_i,\responseLen_i)$ for the total KV-cache resource consumed by one uninterrupted successful execution of job $i$.
\end{definition}

In addition to its role in our algorithms, the memory-time area provides a clear lower bound on the optimal total completion time.
The intuition is that the system can supply at most $\kvmem$ units of memory-time area per round, so the aggregate area of all jobs must be paid for in time; placing smaller-area jobs earlier minimizes the resulting sum of completion times.

\begin{lemma}[Area-order lower bound]
\label{lem:prelim:area-lower-bound}
For every $\SubJobSet\subseteq\JobSet$, if we order the jobs as $\pi(1),\ldots,\pi(\setsize{\SubJobSet})$ so that
$
    \AreaOf{\pi(1)}
    \le \cdots \le
    \AreaOf{\pi(\setsize{\SubJobSet})},
$
then
\begin{align*}
    \OptOf{\SubJobSet}
    \ge
    \frac{1}{\kvmem}
    \sum_{r=1}^{\setsize{\SubJobSet}}
    \left(\setsize{\SubJobSet}-r+1\right)
    \AreaOf{\pi(r)} .
\end{align*}
\end{lemma}
\begin{proof}
Consider any clairvoyant feasible schedule for $\SubJobSet$ and let $\sigma(1),\ldots,\sigma(\setsize{\SubJobSet})$ be the jobs ordered by completion time in that schedule.
By the time job $\sigma(r)$ completes, the schedule must have supplied the successful memory-time area of the first $r$ completed jobs.
Since the system supplies at most $\kvmem$ units of memory-time area per round,
\begin{align*}
    \C_{\sigma(r)}
    \ge
    \frac{1}{\kvmem}
    \sum_{q=1}^{r}
    \AreaOf{\sigma(q)} .
\end{align*}
Summing over $r$ and exchanging the order of summation gives
\begin{align*}
    \sum_{r=1}^{\setsize{\SubJobSet}} \C_{\sigma(r)}
    \ge
    \frac{1}{\kvmem}
    \sum_{q=1}^{\setsize{\SubJobSet}}
    \bigl(\setsize{\SubJobSet}-q+1\bigr)\,
    \AreaOf{\sigma(q)} .
\end{align*}
The right-hand side is minimized when the areas are ordered nondecreasingly, i.e., when $\sigma=\pi$, which yields the claimed bound.
\end{proof}

We also use the following two standard lower bounds.

\begin{lemma}[Deletion bound]
\label{lem:prelim:deletion-bound}
If $\C_i^\star$ is the completion time of job $i$ in an optimal clairvoyant schedule for the full instance, then for any disjoint subsets $\JobSet_1,\ldots,\JobSet_k\subseteq\JobSet$,
\begin{align*}
    \sum_{\ell=1}^k \OptOf{\JobSet_\ell}
    \le
    \sum_{\ell=1}^k\sum_{i\in\JobSet_\ell}\C_i^\star
    \le
    \OPT (\JobSet) .
\end{align*}
\end{lemma}
\begin{proof}
Restricting the optimal full-instance schedule to each subset $\JobSet_\ell$ yields a feasible schedule for $\JobSet_\ell$ with total completion time at most $\sum_{i\in\JobSet_\ell}\C_i^\star$, so $\OptOf{\JobSet_\ell}\le\sum_{i\in\JobSet_\ell}\C_i^\star$ for each $\ell$; summing over $\ell$ gives the first inequality.
The second inequality holds because $\sum_{i\in\JobSet}\C_i^\star=\OPT(\JobSet)$ by definition.
\end{proof}

\begin{lemma}[Processing lower bound]
\label{lem:prelim:processing-lower-bound}
For every $\SubJobSet\subseteq\JobSet$,
\begin{align*}
    \OptOf{\SubJobSet}
    \ge
    \sum_{i\in\SubJobSet}\responseLen_i .
\end{align*}
\end{lemma}
\begin{proof}
Each job $i\in\SubJobSet$ must receive $\responseLen_i$ units of processing and can receive at most one unit per round, so it cannot complete before time $\responseLen_i$; summing over all jobs in $\SubJobSet$ gives the claimed bound.
\end{proof}
 \section{From Area Order to Three Scheduling Cases}
\label{sec:regimes}

In the previous section we derived the area-order lower bound, which gives a tractable benchmark for the optimal clairvoyant completion time.
A natural next step is to ask whether a simple scheduling policy can simultaneously achieve high memory utilization and a completion order consistent with this lower bound.
In this section we show that, outside the large-memory regime, no single priority rule can satisfy both requirements.
This structural impossibility motivates our main algorithmic idea: a decomposition of the job set into three geometric regimes that are handled by specialized sub-schedulers.

The area lower bound in \Cref{sec:preliminaries} suggests that two properties are sufficient for a constant approximation in the clairvoyant setting:
\begin{itemize}
    \item \textbf{(High memory utilization)} the schedule should keep a constant fraction of the memory occupied; and
    \item \textbf{(Area-order completion)} the completion order should be close to nondecreasing memory-time area order.
\end{itemize}
High memory utilization ensures that the system delivers $\Omega(\kvmem)$ units of memory-time area in every round, so the total area of all jobs is paid for within a constant factor of the time needed in the optimal schedule.
Area-order completion ensures that the schedule respects the ordering that minimizes the weighted sum in the area lower bound (\Cref{lem:prelim:area-lower-bound}).
Together, they guarantee that the total completion time is bounded by a constant multiple of the lower bound, yielding a constant approximation.

In the clairvoyant large-memory setting where each job satisfies $\initialLen_i+\responseLen_i=o(\kvmem)$, these two properties can be achieved by area-order greedy, which gives a $3+o(1)$ approximation; see \Cref{cor:clairvoyant-large-memory} in \Cref{sec:clairvoyant-large-memory}.

\xhdr{The failure of a single priority rule.}
Outside the large-memory setting, a single global priority order cannot simultaneously preserve an area-consistent completion order and high memory utilization.

\begin{proposition}
    \label{prop:single-rule-fails}
    No clairvoyant algorithm that follows a single priority rule achieves
    a constant approximation.
    An algorithm follows a single priority rule if it fixes a priority order
    $\prec$ on job types $(\initialLen,\responseLen)$ and never runs a
    lower-priority job while an unfinished higher-priority job is feasible.
\end{proposition}

\begin{proof}
Fix a large integer $L$ and set $\kvmem=L^3+1$.
Consider two job types
\[
    X=(1,L),
    \qquad
    Y=(L^3,1).
\]
Type $X$ has area $\Theta(L^2)$ and type $Y$ has area $\Theta(L^3)$.
Since a type-$Y$ job occupies $L^3+1=\kvmem$ memory in its only active round,
it cannot share any round with another job. Below we consider two possible cases of the priority rule.

\smallskip
\noindent\emph{Case 1: $X\prec Y$.}
    In this case, 
    take one type-$X$ job and $K=\floor{\sqrt{L}}$ type-$Y$ jobs.
    The completion order follows nondecreasing area,
    but memory utilization is at most $(L+1)/\kvmem=o(1)$ while the type-$X$ job
    runs alone for $L$ rounds, blocking all type-$Y$ jobs.
    Thus $\ALG=\Omega(KL)=\Omega(L^{3/2})$.
    Running the type-$Y$ jobs first gives
    $\OPT=O(K^2+L)=O(L)$, hence a ratio of $\Omega(\sqrt{L})$.

\smallskip
\noindent\emph{Case 2: $Y\prec X$.}
    In this case, 
    take $N=\kvmem/(L+1)=L^2-L+1$ type-$X$ jobs
    and $K=\floor{L^{3/2}}$ type-$Y$ jobs.
    Memory is fully utilized during type-$Y$ rounds,
    but the completion order reverses the area ranking:
    all $N$ small-area type-$X$ jobs are blocked until the $K$
    serial type-$Y$ jobs complete.
    Thus $\ALG\ge NK=\Omega(L^{7/2})$.
    Running all type-$X$ jobs together first
    (they fit since $N(L+1)\le\kvmem$) and then the type-$Y$ jobs
    serially gives $\OPT=O(L^3)$, hence a ratio of $\Omega(\sqrt{L})$.

\smallskip
\noindent
Combining the two cases together, since $L=\Theta(\kvmem^{1/3})$ is unbounded, no such single priority rule gives a constant approximation.
\end{proof}
The proposition rules out any single priority order, including natural global orders such as response-length order and area order.
In the non-clairvoyant setting, the obstruction is further compounded because response lengths are unknown and a killed attempt loses all accumulated progress.
The two cases impose opposite pressures on a non-clairvoyant scheduler.
On one hand, small-prompt jobs must sometimes yield memory to large-prompt jobs; otherwise short large-prompt jobs may remain undiscovered.
On the other hand, large-prompt jobs cannot be run too aggressively, because one such attempt may already fill the cache and block many small-prompt jobs.
The fundamental obstacle is serving both geometries without letting either monopolize the memory budget.
This points to a decomposition into geometric regimes: we first ask whether each regime, taken independently, admits a constant-factor guarantee.\footnote{The hard instance in \Cref{prop:single-rule-fails} mixes large prompt-heavy jobs with response-heavy jobs. In the two cases of the proposition, the failure of the constant ratio comes from violating high-memory utilization and area-order completion, respectively.}

\xhdr{\PromptHeavyCap jobs.}
For jobs whose prompt length dominates the response length, the growing memory use is close to a fixed width.
If $\responseLen_i\le\initialLen_i$, then during a successful attempt,
\begin{equation*}
    \initialLen_i+1
    \le
    \initialLen_i+\curProgress_{i,t}+1
    \le
    2\initialLen_i .
\end{equation*}
Hence the dynamic KV-cache memory use is within a constant factor of a fixed rectangle width.
This motivates the rectangle strip abstraction: job $i$ is represented by a fixed rectangle width $\rectWidth_i$, and the scheduler packs active rectangle widths under the same memory budget.

The fixed-width view gives good utilization when rectangle widths are uniformly small or uniformly large.
\begin{enumerate}
    \item
    Suppose every relevant rectangle width is at most $\WidthUpperBound$.
    Then utilization comes from greedy packing: if a waiting job cannot be inserted, then the active rectangles already use more than $\kvmem-\WidthUpperBound$ memory.
    Therefore the schedule has high utilization whenever $\kvmem-\WidthUpperBound = \Omega(\kvmem)$.

    \item
    Suppose every relevant rectangle width is at least $\WidthLowerBound$.
    Then any non-idle round uses at least $\WidthLowerBound$ memory.
    Therefore any work-conserving schedule has constant utilization whenever $\WidthLowerBound=\Omega(\kvmem)$.
\end{enumerate}
Mixing these widths in one schedule can destroy the utilization lower bound.
We therefore separate jobs by the prompt length first, using a prompt cutoff $\PromptCutoff \in (0,\kvmem/2)$: \emph{large} jobs with $\initialLen_i>\PromptCutoff$ are separated first, and the other \emph{small} jobs have bounded rectangle widths under the fixed-width view when $\responseLen_i \le \initialLen_i$.

\xhdr{\ResponseHeavyCap jobs.}
When $\responseLen_i > \initialLen_i$ instead, the memory footprint grows well beyond the prompt length and the fixed-width view no longer applies.
For these jobs, the area is controlled by the response length:
\begin{equation*}
    \frac{\responseLen_i^2}{2}
    \le
    \AreaOf{i}
    \le
    2\responseLen_i^2 .
\end{equation*}
Thus response-length order and area order agree up to constant factors.

\xhdr{Three scheduling cases.}
The above discussions give the following three regimes:
\begin{itemize}
\item \textbf{(Large jobs)}
    $\initialLen_i > \PromptCutoff$. These jobs are so wide that each one nearly fills the memory budget alone.
    \item \textbf{(Small prompt-heavy jobs)} $\initialLen_i \le \PromptCutoff$ and $\responseLen_i \le \initialLen_i$. The response is no longer than the prompt, so the memory footprint stays within a constant factor of the prompt length.
    \item \textbf{(Small response-heavy jobs)} $\initialLen_i \le \PromptCutoff$ and $\responseLen_i > \initialLen_i$. The memory footprint grows well beyond the prompt, and the area is dominated by the response length.
\end{itemize}

Each regime, taken in isolation, admits a constant-competitive non-clairvoyant algorithm. Specifically, \Cref{sec:rectangle-strip,sec:response-heavy} give constant-competitive algorithms for these regimes (\Cref{prop:rectangle-large-guarantee,prop:rectangle-prompt-heavy-guarantee} and \Cref{thm:response-heavy-ratio}).
The remaining challenge is that the non-clairvoyant scheduler does not know which regime a job belongs to, because the distinction between prompt-heavy and response-heavy depends on the unknown response length.
We explain how to resolve this difficulty when we discuss the meta-scheduler below.

If the three-way partition were known in advance, one could simply run each sub-scheduler on a dedicated replica (a three-fold resource augmentation) and obtain an $O(1)$-competitive guarantee immediately.
Without resource augmentation, the black-box meta-scheduler of \Cref{sec:meta-scheduling} time-shares a single memory budget across the three sub-schedulers, losing only a constant factor.
In the non-clairvoyant setting, the additional difficulty is that the partition is not known: the scheduler cannot distinguish small \PromptHeavy jobs from small \ResponseHeavy jobs without executing them.
We resolve this by routing jobs dynamically as their response lengths reveal themselves, using the routing meta-scheduler described in \Cref{sec:meta-scheduling}.

 \section{Non-Clairvoyant Rectangle Strip Scheduling}
\label{sec:rectangle-strip}
This section analyzes two of the three branches: \LargeBranch and \PromptBranch. These two branches share the following useful property: for each job, its memory growth can be bounded by a constant factor of its prompt length. In particular, in \LargeBranch, the rectangle width is at most $\kvmem$, which is $O(\initialLen_i)$ when the prompt cutoff $\PromptCutoff$ satisfies $\PromptCutoff=\Theta(\kvmem)$. In \PromptBranch, the width is at most $2 \initialLen_i$.

This property allows us to reduce the analysis of these two branches to a non-clairvoyant rectangle strip scheduling problem. In this problem, each job is represented by a rectangle with a known fixed width and an unknown processing time. The clairvoyant counterpart of this problem has been studied in \citet{CCZ24}. To the best of our knowledge, the non-clairvoyant version has not been investigated before. We prove an $O(1)$-competitive algorithm for this problem, which may also be of independent interest. Next, we define the non-clairvoyant rectangle strip scheduling problem.

\subsection{The Rectangle Strip Scheduling Problem}
\label{subsec:rectangle-strip-problem}

An instance of the rectangle strip problem consists of a job set $\JobSet$, a memory budget $\kvmem$, a known width $\rectWidth_i \leq \kvmem$, and an unknown processing time $\processingTime_i$ for each job $i\in\JobSet$.
At every round, the total width of the active attempts must be at most $\kvmem$.
An active attempt of job $i$ uses width $\rectWidth_i$ and completes the job after receiving $\processingTime_i$ consecutive processing rounds.
If an attempt is stopped earlier, its progress is lost.
The processing time $\processingTime_i$ is not known to the scheduler.
Let $\OptRectOf{\JobSet}$ denote the optimal clairvoyant total completion time for this rectangle strip instance.

\subsection{The Rectangle Strip Scheduler and Analysis}
\label{subsec:rectangle-strip-scheduler}

The scheduler \ALGRect processes each job through a sequence of capped attempts with geometrically increasing length caps.
Since $\processingTime_i$ is unknown, this prevents the schedule from spending too much time on any single job.
If the job does not complete within this cap, the attempt expires and is killed.
The next attempt restarts the job from zero progress with the doubled cap.
Thus the level-$\Index$ attempt of job $i$ can run for at most $2^\Index$ rounds.

Based on this doubling rule, \ALGRect is a greedy allocator over attempts.
For attempt $(i,\Index)$, define its area budget by
\begin{equation*}
    \BudgetOf{i}{\Index}
    \triangleq
    \rectWidth_i2^\Index .
\end{equation*}
This is a doubling estimate of the true rectangle area $\RectAreaOf{i}\triangleq\rectWidth_i\processingTime_i$.
A pending queue $\PendingQueue$ is maintained, containing all job attempts waiting to be started.
At every event time, the scheduler checks whether the pending attempt with the minimum area budget ($\min\PendingQueue$) can fit into the remaining memory.
If so, it starts the attempt and repeats the scan; otherwise, the scan stops.

The critical events are time $0$, a job completion, and an attempt expiration.
Without loss of generality, simultaneous events in \Cref{alg:rectangle-strip} are processed one at a time under an arbitrary fixed tie-breaking order:
each event time after time $0$ contains only one completion or one expiration.
The pseudocode is given in \Cref{alg:rectangle-strip}.

\begin{algorithm2e}
\caption{Budget-Doubling Greedy for Rectangle Strips ($\ALGRect$)}
\label{alg:rectangle-strip}
\SetKw{Break}{break}
\KwIn{Rectangle job set $\JobSet$, widths $\{\rectWidth_i\}_{i\in\JobSet}$, memory $\kvmem$.}
$\PendingQueue\leftarrow\{(i,0):i\in\JobSet\}$\tcp*{attempt $(i,\Index)$: length cap $2^\Index$, budget $\BudgetOf{i}{\Index}=\rectWidth_i2^\Index$}

\For{each event in order: time $0$, a completion, or an expiration}{
    \uIf{the current event is the completion of an active attempt}{
        Finalize that completed attempt\;
    }
    \ElseIf{the current event is the expiration of an active attempt $(i,\Index)$}{
        Kill $(i,\Index)$ and insert $(i,\Index+1)$ into $\PendingQueue$\;
    }
    \tcp{Greedy admission scan until no more jobs can be accepted.}
    \While{$\PendingQueue\neq\varnothing$}{
        $(i,\Index)\leftarrow\min \PendingQueue$\tcp*{the minimum budget attempt}
        \If{$(i,\Index)$ does not fit in the remaining memory}{
            \Break\;
        }
        Remove $(i,\Index)$ from $\PendingQueue$ and start it\;
    }
}
\end{algorithm2e}

\xhdr{Proof overview.}
In the analysis, we decompose the completion time of each job into its active time and queue waiting time.
The active time is immediate: the total active time of job $i$ over all its attempts is at most $4\processingTime_i$; it remains to bound the queue waiting time.
To relate the waiting time to the optimal solution, we use a memory-area argument, extending the clairvoyant argument from \Cref{sec:regimes}. This argument relies on the following two properties maintained by \Cref{alg:rectangle-strip}:
\begin{enumerate}[label=(\roman*)]
    \item attempts are processed in order of nearly nondecreasing rectangle area budget; and
    \item every round with a nonempty queue uses a constant fraction of the memory.
\end{enumerate}

Property (ii) turns queue waiting time into an area charge:
whenever an attempt waits, the greedy scan and the width conditions ensure that a constant fraction of the memory is occupied by active attempts.
Property (i) makes this charge compatible with the area order:
the active area charged to a waiting attempt must come only from attempts with no larger budget.
This condition lets us bound the total waiting contribution by the area-order lower bound.

The difficulty in proving property (i) is that attempts are created over time.
Although the greedy scan always starts a minimum-budget pending attempt when it fits, a replacement attempt $(i,\Index)$ appears only when $(i,\Index-1)$ expires.
While $(i,\Index-1)$ is active, the scan may have admitted other attempts whose budgets are larger than the budget of the replacement.
If $(i,\Index)$ remained pending, those active attempts would have a larger budget than a pending attempt, violating the desired order.

We rule out this situation using a crucial property of the rectangle strip model:
the width of each job is fixed across all attempts, even though the length cap $2^\Index$ increases with the level $\Index$.
When $(i,\Index-1)$ expires, it releases a block of memory of width $\rectWidth_i$.
The replacement $(i,\Index)$ has the same width $\rectWidth_i$, so this released memory is enough to run the replacement.
If the replacement budget is small enough to threaten the budget order, then it has priority over every older pending attempt.
The following greedy scan therefore starts the replacement immediately, before any older pending attempt can use the released memory.
The approximate budget order is formalized by queue monotonicity (\Cref{lem:rect-queue-monotonicity}): every pending attempt has budget at least that of every active attempt in each settled state.

This fixed-width argument is specific to the rectangle strip regime, which captures both \LargeBranch and \PromptBranch. In contrast, in the \ResponseHeavy regime, the effective width of a job may continue to grow, and therefore the reduction to fixed-width rectangle scheduling no longer applies. For this reason, \Cref{sec:response-heavy} uses a different approach to maintain a good competitive ratio. This separation explains why the branches require different scheduling policies.

To precisely state the property of our algorithm, we index completion and expiration events by $\ell=1,2,\dots$, and let $\ell=0$ denote the initial event at time $0$.
We call a state \emph{settled} when the greedy admission scan has stopped.
Let $S_\ell$ denote the settled state at the end of the outer \textbf{for} loop in \Cref{alg:rectangle-strip} after event $\ell$, and let $\PendingQueue(S_\ell)$ denote the pending queue in $S_\ell$.

\begin{lemma}[Queue monotonicity]
\label{lem:rect-queue-monotonicity}
Consider the minimum budget among all pending attempts at a settled state:
\begin{equation*}
    \min_{(i,\Index)\in\PendingQueue(S_\ell)}\BudgetOf{i}{\Index}.
\end{equation*}
This quantity is nondecreasing with $\ell=0,1,2,\dots$, with the convention that its value is $+\infty$ when $\PendingQueue(S_\ell)=\varnothing$.
Moreover, at any settled state $S_\ell$, every active attempt $(i,\Index)$ and every pending attempt $(j,\OtherIndex)$ satisfy
\begin{equation*}
    \BudgetOf{i}{\Index}
    \le
    \BudgetOf{j}{\OtherIndex}.
\end{equation*}
\end{lemma}

\begin{proof}
We prove the two claims by induction on $\ell$.

For the base case $\ell=0$, all attempts initially have level $0$ and are placed in the pending queue.
The greedy admission scan repeatedly starts a minimum-budget pending attempt until the queue is empty or the minimum-budget pending attempt no longer fits.
Thus, at the settled state $S_0$, every active attempt has budget at most every remaining pending attempt.

Now fix $\ell\ge 1$, assume both claims hold at settled state $S_{\ell-1}$, and consider event $\ell$.
We first show that the minimum pending budget does not decrease, by distinguishing two kinds of events.

\xhdr{Completion.}
Suppose event $\ell$ is a completion.
No new attempt is inserted before the greedy admission scan.
The scan only removes attempts from the pending queue and starts them.
Thus the minimum pending budget cannot decrease.

\xhdr{Expiration.}
Suppose event $\ell$ is the expiration of an active attempt $(i,r)$.
The scheduler kills this attempt and inserts its replacement $(i,r+1)$.
The replacement has the same width $\rectWidth_i$ and a larger budget.

If
\begin{equation*}
    \BudgetOf{i}{r+1}
    \ge
    \min_{(h,q)\in\PendingQueue(S_{\ell-1})}\BudgetOf{h}{q},
\end{equation*}
then inserting $(i,r+1)$ does not decrease the minimum pending budget.
The following greedy admission scan only removes pending attempts and starts them, so the minimum pending budget at $S_\ell$ is at least the minimum pending budget at $S_{\ell-1}$.

It remains to consider the case
\begin{equation*}
    \BudgetOf{i}{r+1}
    <
    \min_{(h,q)\in\PendingQueue(S_{\ell-1})}\BudgetOf{h}{q}.
\end{equation*}
Then $(i,r+1)$ has strictly smaller budget than every older pending attempt.
The expired attempt $(i,r)$ has just released width $\rectWidth_i$, and the replacement $(i,r+1)$ has the same width.
Therefore $(i,r+1)$ fits immediately after the expiration.
Since it is the minimum-budget pending attempt, the greedy admission scan starts it before any older pending attempt.
Thus $(i,r+1)$ cannot remain pending in $S_\ell$.
All attempts that remain pending in $S_\ell$ are older pending attempts from $\PendingQueue(S_{\ell-1})$, except that some of them may have been removed by the scan.
Therefore the minimum pending budget at $S_\ell$ is again at least the minimum pending budget at $S_{\ell-1}$.

\xhdr{Budget order.}
In both cases, the minimum pending budget does not decrease from $S_{\ell-1}$ to $S_\ell$.
An active attempt in $S_\ell$ either survived from $S_{\ell-1}$ or was started during the greedy admission scan after event $\ell$.
If it survived from $S_{\ell-1}$, the induction hypothesis bounds its budget by the minimum pending budget at $S_{\ell-1}$; since the minimum pending budget does not decrease, its budget is also at most every pending budget at $S_\ell$.
If it was started during the scan, then it was a minimum-budget pending attempt at the moment it was started.
Because the scan always removes a current minimum-budget pending attempt, any attempt started by the scan has budget at most every attempt that remains pending when the scan stops.
Thus every active attempt in $S_\ell$ has budget at most every pending attempt in $S_\ell$.
\end{proof}

Write $\RectAreaOf{i}\triangleq\rectWidth_i\processingTime_i$ for the rectangle area of job $i$.
Let
$
    \Index_i
    \triangleq
    \min\{\Index:2^\Index\ge \processingTime_i\}
$
be the first successful level.
Then
\begin{equation*}
    \RectAreaOf{i}
    \le
    \BudgetOf{i}{\Index_i}
    <
    2\RectAreaOf{i} .
\end{equation*}

\begin{lemma}[Per-job completion-time control]
\label{lem:rect-per-job}
Assume the rectangle instance on $\JobSet$ has waiting utilization $U$: in every settled state where $\PendingQueue \neq \varnothing$, the active attempts use at least $U$ total memory.
For every job $i\in\JobSet$,
\begin{equation*}
    \C_i^{\ALGRect}
    \le
    4\processingTime_i
    +
    \frac{4}{U}
    \sum_{j\in\JobSet}\min\{\RectAreaOf{i},\RectAreaOf{j}\}.
\end{equation*}
\end{lemma}

\begin{proof}
Let $\WaitingTime_i$ be the total time during which there exists an attempt of job $i$ pending before it starts.
The active time is bounded by the geometric caps:
\begin{equation*}
    \sum_{\Index=0}^{\Index_i}2^\Index
    <
    2^{\Index_i+1}
    \le
    4\processingTime_i .
\end{equation*}
Hence
\begin{equation*}
    \C_i^{\ALGRect}
    \le
    4\processingTime_i+\WaitingTime_i .
\end{equation*}
Consider any waiting round counted by $\WaitingTime_i$.
Let $(i,\OtherIndex)$ be the pending attempt of job $i$ in that round.
By \Cref{lem:rect-queue-monotonicity}, every active attempt $(j,\Index)$ satisfies
\begin{equation*}
    \BudgetOf{j}{\Index}
    \le
    \BudgetOf{i}{\OtherIndex}
    \le
    \BudgetOf{i}{\Index_i}.
\end{equation*}
The utilization assumption gives at least $U$ active memory in this round, all from attempts of budget at most $\BudgetOf{i}{\Index_i}$.
Therefore
\begin{equation*}
    U \WaitingTime_i
    \le
    \sum_{j\in\JobSet}
    \sum_{\substack{0\le \OtherIndex\le \Index_j\\ \BudgetOf{j}{\OtherIndex}\le \BudgetOf{i}{\Index_i}}}
    \BudgetOf{j}{\OtherIndex}.
\end{equation*}
For each fixed job $j$,
\begin{equation*}
    \sum_{\substack{0\le \OtherIndex\le \Index_j\\ \BudgetOf{j}{\OtherIndex}\le \BudgetOf{i}{\Index_i}}}
    \BudgetOf{j}{\OtherIndex}
    <
    2\min\{\BudgetOf{i}{\Index_i},\BudgetOf{j}{\Index_j}\}
    \le
    4\min\{\RectAreaOf{i},\RectAreaOf{j}\},
\end{equation*}
where the last inequality uses the successful budget bound above.
Combining the previous displays gives
\begin{equation*}
    \WaitingTime_i
    \le
    \frac{4}{U}
    \sum_{j\in\JobSet}\min\{\RectAreaOf{i},\RectAreaOf{j}\}.
\end{equation*}
Together with the active-time bound above, this proves the lemma.
\end{proof}

\begin{lemma}[Rectangle lower bounds]
\label{lem:rect-lower-bounds}
For every rectangle strip instance on $\JobSet$,
\begin{equation*}
    \OptRectOf{\JobSet}
    \ge
    \sum_{i\in\JobSet}\processingTime_i
\end{equation*}
and
\begin{equation*}
    \frac{1}{2\kvmem}
    \sum_{i,j\in\JobSet}\min\{\RectAreaOf{i},\RectAreaOf{j}\}
    \le
    \OptRectOf{\JobSet}.
\end{equation*}
\end{lemma}

\begin{proof}
The processing lower bound is immediate.
For the area lower bound, order the jobs so that
\begin{equation*}
    \RectAreaOf{\pi(1)}
    \le
    \RectAreaOf{\pi(2)}
    \le
    \cdots
    \le
    \RectAreaOf{\pi(\setsize{\JobSet})}.
\end{equation*}
The same area-order lower-bound argument (\Cref{lem:prelim:area-lower-bound}) applied to rectangle areas gives
\begin{equation*}
    \OptRectOf{\JobSet}
    \ge
    \frac1\kvmem
    \sum_{r=1}^{\setsize{\JobSet}}
    \left(\setsize{\JobSet}-r+1\right)\RectAreaOf{\pi(r)}.
\end{equation*}
Also,
\begin{equation*}
    \sum_{i,j\in\JobSet}\min\{\RectAreaOf{i},\RectAreaOf{j}\}
    =
    \sum_{r=1}^{\setsize{\JobSet}}
    \left(2(\setsize{\JobSet}-r)+1\right)\RectAreaOf{\pi(r)}
    \le
    2
    \sum_{r=1}^{\setsize{\JobSet}}
    \left(\setsize{\JobSet}-r+1\right)\RectAreaOf{\pi(r)}.
\end{equation*}
The claim follows.
\end{proof}

\begin{lemma}[Width utilization]
\label{lem:rect-width-range}
Suppose every job satisfies
\begin{equation*}
    \WidthLowerBound
    \le
    \rectWidth_i
    \le
    \WidthUpperBound.
\end{equation*}
Then the rectangle instance has waiting utilization
\begin{equation*}
    U
    =
    \max\{\WidthLowerBound,\kvmem-\WidthUpperBound\}.
\end{equation*}
\end{lemma}

\begin{proof}
Consider a settled state in which an attempt is waiting.
Let $m$ be the total active memory.
Let $(i,\Index)=\min \PendingQueue$.
Since $(i,\Index)$ does not fit and has width at most $\WidthUpperBound$, we have
\begin{equation*}
    m
    >
    \kvmem-\WidthUpperBound .
\end{equation*}
Also, the active set is nonempty; otherwise $(i,\Index)$ would fit because $\rectWidth_i\le\kvmem$.
Since every active attempt has width at least $\WidthLowerBound$, we also have
\begin{equation*}
    m
    \ge
    \WidthLowerBound .
\end{equation*}
Thus the instance has waiting utilization $\max\{\WidthLowerBound,\kvmem-\WidthUpperBound\}$.
\end{proof}

\begin{theorem}[Rectangle strip guarantee]
\label{thm:rectangle-strip}
Suppose every job satisfies
\begin{equation*}
    \WidthLowerBound
    \le
    \rectWidth_i
    \le
    \WidthUpperBound.
\end{equation*}
Then
\begin{equation*}
    \ALGRect(\JobSet)
    \le
    \left(
        4
        +
        \frac{8\kvmem}
        {\max\{\WidthLowerBound,\kvmem-\WidthUpperBound\}}
    \right)
    \OptRectOf{\JobSet}.
\end{equation*}
\end{theorem}

\begin{proof}
Let $U=\max\{\WidthLowerBound,\kvmem-\WidthUpperBound\}$.
By \Cref{lem:rect-width-range}, the rectangle instance has waiting utilization $U$.
Summing \Cref{lem:rect-per-job} over all jobs gives
\begin{equation*}
    \ALGRect(\JobSet)
    \le
    4\sum_{i\in\JobSet}\processingTime_i
    +
    \frac4U
    \sum_{i,j\in\JobSet}\min\{\RectAreaOf{i},\RectAreaOf{j}\}.
\end{equation*}
The theorem follows from \Cref{lem:rect-lower-bounds}.
\end{proof}

\begin{remark}
For a general rectangle strip instance with no width range assumption, the same scheduler can be applied to two width classes and combined by a black-box doubling meta-scheduler defined in \Cref{sec:meta-known-partition}.
This gives a constant-competitive algorithm for the full rectangle strip problem; the complete statement and proof are in \Cref{sec:rectangle-strip-complete}.
\end{remark}

\subsection{Applying to Homogeneous KV-Cache Regimes}
\label{subsec:rectangle-branches}

The \LargeBranch and the \PromptBranch are both applications of \ALGRect to homogeneous KV-cache regimes.
When every job satisfies the same regime condition, the corresponding branch is an ordinary rectangle strip schedule with a direct completion-time guarantee.
They are defined as follows.
\begin{itemize}
    \item \textbf{\LargeBranchCap.} Use $\ALGRect$ with rectangle width $\rectWidth_i=\kvmem$ and processing time $\processingTime_i=\responseLen_i$.
    \item \textbf{\PromptBranchCap.} Use $\ALGRect$ with rectangle width $\rectWidth_i=2\initialLen_i$ and processing time $\processingTime_i=\responseLen_i$.
\end{itemize}

\begin{proposition}[Large homogeneous guarantee]
\label{prop:rectangle-large-guarantee}
If every job satisfies $\initialLen_i>\PromptCutoff$, then the \LargeBranch schedule has total completion time at most
\begin{equation*}
    \left(4+\frac{8\kvmem}{\PromptCutoff}\right)
    \OptOf{\JobSet}.
\end{equation*}
\end{proposition}

\begin{proof}
By per-job feasibility, every rectangle completion is true completion in the original KV-cache instance.
Taking $\WidthLowerBound=\WidthUpperBound=\kvmem$ satisfies the width assumptions in \Cref{thm:rectangle-strip}.
Also, for every job $i$,
\begin{equation*}
    \RectAreaOf{i}
    =
    \kvmem\responseLen_i
    \le
    \frac{\kvmem}{\PromptCutoff}
    \initialLen_i\responseLen_i
    \le
    \frac{\kvmem}{\PromptCutoff}
    \AreaOf{i}.
\end{equation*}
As in the proof of \Cref{thm:rectangle-strip}, by \Cref{lem:rect-per-job} and \Cref{lem:rect-width-range},
\begin{equation*}
    \ALGRect(\JobSet)
    \le
    4\sum_{i\in\JobSet}\responseLen_i
    +
    \frac4M
    \sum_{i,j\in\JobSet}\min\{\RectAreaOf{i},\RectAreaOf{j}\}
    \le
    4\sum_{i\in\JobSet}\responseLen_i
    +
    \frac4M \frac{M}{\PromptCutoff}
    \sum_{i,j\in\JobSet}\min\{\AreaOf{i},\AreaOf{j}\}
\end{equation*}
As in the proof of \Cref{lem:rect-lower-bounds}, by the area-order lower bound argument in \Cref{lem:prelim:area-lower-bound},
\begin{equation*}
    \sum_{i\in\JobSet}\responseLen_i
    \le
    \OptOf{\JobSet},
\end{equation*}
and
\begin{equation*}
    \frac{1}{2\kvmem}
    \sum_{i,j\in\JobSet}\min\{\AreaOf{i},\AreaOf{j}\}
    \le
    \OptOf{\JobSet}.
\end{equation*}
Thus we conclude
\begin{equation*}
    \ALGRect(\JobSet)
    \le
    \left(4+\frac{8\kvmem}{\PromptCutoff}\right)
    \OptOf{\JobSet}.
\end{equation*}
\end{proof}

\begin{proposition}[\PromptHeavyCap homogeneous guarantee]
\label{prop:rectangle-prompt-heavy-guarantee}
If every job satisfies $\initialLen_i\le\PromptCutoff$ and $\responseLen_i\le\initialLen_i$, then the \PromptBranch schedule has total completion time at most
\begin{equation*}
    \left(4+\frac{16\kvmem}{\kvmem-2\PromptCutoff}\right)
    \OptOf{\JobSet}.
\end{equation*}
\end{proposition}

\begin{proof}
Under the \PromptHeavy condition, every rectangle completion is true completion in the original KV-cache instance.
Every rectangle width is at most $2\PromptCutoff$.
Taking $\WidthLowerBound=0$ and $\WidthUpperBound=2\PromptCutoff$ satisfies the width assumptions in \Cref{thm:rectangle-strip}.
Moreover,
\begin{equation*}
    \RectAreaOf{i}
    =
    2\initialLen_i\responseLen_i
    \le
    2\AreaOf{i}.
\end{equation*}
The same argument as Large branch (\Cref{prop:rectangle-large-guarantee}) gives
\begin{align*}
    \ALGRect(\JobSet)
    &\le 4\sum_{i\in\JobSet}\responseLen_i
      + \frac{4}{M-2\PromptCutoff}
        \sum_{i,j\in\JobSet}\min\{\RectAreaOf{i},\RectAreaOf{j}\} \\
    &\le 4\sum_{i\in\JobSet}\responseLen_i
      + \frac{8}{M-2\PromptCutoff}
        \sum_{i,j\in\JobSet}\min\{\AreaOf{i},\AreaOf{j}\} \\
    &\le \left(4+\frac{16\kvmem}{\kvmem-2\PromptCutoff}\right)
        \OptOf{\JobSet}.
\end{align*}
\end{proof}

\begin{remark}
For $\PromptCutoff=\kvmem/4$, the constants in \Cref{prop:rectangle-large-guarantee,prop:rectangle-prompt-heavy-guarantee} are both equal to $36$.
\end{remark}
 \section{Area-Geometric Slicing for \ResponseHeavyCap Jobs}
\label{sec:response-heavy}

\ResponseHeavyCap jobs create a different obstacle from the rectangle strip setting of \Cref{sec:rectangle-strip}.
Their memory footprint can grow far beyond the prompt length, so the fixed-width queue-monotonicity argument from \Cref{sec:rectangle-strip} no longer applies.

This section uses a structural fact absent from the \PromptHeavy regime: response length determines area up to constant factors.
Every job $i$ considered in this section satisfies
\begin{align*}
    \initialLen_i<\responseLen_i,
    \qquad
    \initialLen_i+\responseLen_i\le\kvmem .
\end{align*}
Under this condition, memory-time area and the squared response length are equivalent up to constants:
\begin{align*}
    \frac{\responseLen_i^2}{2}
    \le
    \AreaOf{i}
    \le
    2\responseLen_i^2 .
\end{align*}
Thus ordering jobs by response caps also orders their area scales within constant factors.
This equivalence permits each phase to replace the true prompts with a common prompt upper bound with only a constant-factor loss.

The scheduler $\ALGResp$ uses geometric slicing to exploit this equivalence.
The geometric slicing approach of~\cite{FYZ26} was developed for the identical-prompt setting.
Although prompts may differ across jobs here, the area equivalence above allows each phase of $\ALGResp$ to impose a common response cap and replace heterogeneous prompts with a common proxy prompt.
The resulting proxy jobs are scheduled by the $\SPS$ subroutine described below.
This construction realizes the two sufficient conditions from \Cref{sec:regimes}: $\SPS$ provides high memory utilization within each phase, and the geometric response caps ensure that the completion order is consistent with nearly nondecreasing area.

\xhdr{The SPS subroutine.}
The Staggered Pipeline Schedule ($\SPS$) is a deterministic subroutine from \citet{FYZ26} for scheduling an ordered batch of identical jobs.
An $\SPS$ instance consists of $m\ge1$ identical jobs, each with common prompt length $\initialLen\ge0$ and slice length $\sliceLen\ge1$.
Given a degree of parallelism $\paral\ge1$, the schedule assigns the job of rank $j\in\{0,\ldots,m-1\}$ the start offset
\begin{equation*}
    \StartTimeOf{j}
    =
    \floor{\frac{j\sliceLen}{\paral}},
\end{equation*}
and runs that job for $\sliceLen$ rounds.
Consecutive jobs are therefore spaced roughly $\sliceLen/\paral$ rounds apart, so at most $\paral$ jobs overlap at any time.
The key property is efficiency: by staggering the start times, $\SPS$ can choose a nearly optimal parallelism.

\begin{fact}[SPS properties from \citet{FYZ26}]
\label{fact:rh:sps}
For a degree of parallelism $\paral$, slice length $\sliceLen$, and common prompt length $\initialLen$, the peak memory of the infinite $\SPS$ schedule is
\begin{align*}
    \MemPeak(\paral,\sliceLen,\initialLen)
    \triangleq
    \initialLen\paral
    +
    \frac{\sliceLen\paral+\sliceLen+\paral-\gcd(\sliceLen,\paral)}{2}.
\end{align*}
Hence, if $\MemPeak(\paral,\sliceLen,\initialLen)\le\kvmem$, then every finite $\SPS$ prefix is feasible.
In particular, whenever $\sliceLen+\initialLen\le\kvmem$, the maximum feasible degree of parallelism satisfies
\begin{equation*}
    \paral^*(\sliceLen,\initialLen)
    \triangleq
    \max\{\paral\in\naturals^{+}:\MemPeak(\paral,\sliceLen,\initialLen)\le\kvmem\}
    \ge
    \left\lfloor
        \frac{2\kvmem-\sliceLen+1}{2\initialLen+\sliceLen+1}
    \right\rfloor
    \ge 1.
\end{equation*}
\end{fact}

The \ALGResp scheduler (\Cref{alg:response-heavy}) uses geometrically increasing response caps.
Fix a scaling factor $\responseScalar>1$.
Let $\LastRespPhase\triangleq\ceil{\log_{\responseScalar}\kvmem}$.
In phase $\Index$, the response cap is
\begin{equation*}
    \sliceLen_{\Index}
    =
    \min\{\ceil{\responseScalar^{\Index}},\kvmem\}.
\end{equation*}
For each phase $\Index$, define the proxy prompt
\begin{equation*}
    \ProxyPrompt_{\Index}
    =
    \min\{\sliceLen_{\Index},\kvmem-\sliceLen_{\Index}\}
\end{equation*}
and choose the maximum feasible proxy parallelism $\paral_{\Index} = \paral^*(\sliceLen_{\Index}, \ProxyPrompt_{\Index})$.
The \SPS call uses the induced order on its current job set from a deterministic order on the input job set; a job that reaches the cap without completing is killed.

\begin{algorithm2e}
\caption{Response-Geometric SPS ($\ALGResp$)}
\label{alg:response-heavy}
\KwIn{Job set $\JobSet$, known prompt lengths $\{\initialLen_i\}$, memory $\kvmem$, scaling factor $\responseScalar>1$.}
$\UnfinishedJobs\leftarrow\JobSet$\;
\For{$\Index=0,1,\ldots,\LastRespPhase$}{
    Set $\sliceLen_{\Index}\leftarrow\min\{\ceil{\responseScalar^{\Index}},\kvmem\}$\;
    Set $\ProxyPrompt_{\Index}\leftarrow\min\{\sliceLen_{\Index},\kvmem-\sliceLen_{\Index}\}$\;
    Set $\paral_{\Index}\leftarrow \paral^*(\sliceLen_{\Index}, \ProxyPrompt_{\Index})$\;
    $\remainJobs_{\Index}\leftarrow\{i\in\UnfinishedJobs:\initialLen_i\le\sliceLen_{\Index}\}$\;
    Run the $\SPS$ \citep{FYZ26} for prompt $\ProxyPrompt_{\Index}$, slice $\sliceLen_{\Index}$, and parallelism $\paral_{\Index}$ on $\remainJobs_{\Index}$\;
    \tcp{When the proxy slot of job $i$ starts, run the real job $i$.}
    \tcp{Kill job $i$ if it does not complete after $\sliceLen_{\Index}$ rounds.}
    Remove completed jobs from $\UnfinishedJobs$\;
}
\end{algorithm2e}

\begin{lemma}[Feasibility of \ALGResp]
\label{lem:rh:phase-feasible}
Every phase of \Cref{alg:response-heavy} is feasible.
\end{lemma}

\begin{proof}
Fix phase $\Index$ and write $\sliceLen=\sliceLen_{\Index}$ and $\paral=\paral_{\Index}$.
Every scheduled job satisfies $\initialLen_i\le\sliceLen$.

If $\sliceLen>\kvmem/2$, we have
\begin{equation*}
    \MemPeak(2,\sliceLen,\kvmem-\sliceLen)
    =
    2(\kvmem-\sliceLen)
    +
    \frac{2\sliceLen+\sliceLen+2-\gcd(\sliceLen,2)}2
    \ge
    2\kvmem-\sliceLen+1
    > \kvmem,
\end{equation*}
hence $\paral=1$; each job runs sequentially, and per-job feasibility $\initialLen_i+\responseLen_i\le\kvmem$ bounds the memory in every round.

If $\sliceLen\le\kvmem/2$, then $\ProxyPrompt_{\Index}=\sliceLen\ge\initialLen_i$, so the real memory $\initialLen_i+u+1$ at progress $u$ does not exceed the proxy memory $\ProxyPrompt_{\Index}+u+1$; the proxy $\SPS$ schedule is feasible by \Cref{fact:rh:sps} and the choice of $\paral$.
\end{proof}

\begin{lemma}[Completion of \ALGResp]
\label{lem:rh:completion-cap}
Let $\Index$ be the first phase with $\sliceLen_{\Index}\ge\responseLen_i$.
Job $i$ completes in phase $\Index$.
\end{lemma}

\begin{proof}
For any phase $\OtherIndex<\Index$, the cap satisfies $\sliceLen_{\OtherIndex}<\responseLen_i$, so job $i$ receives fewer than $\responseLen_i$ rounds and cannot complete.
In phase $\Index$, response dominance gives $\initialLen_i<\responseLen_i\le\sliceLen_{\Index}$, so job $i$ is eligible.
The phase runs job $i$ for up to $\sliceLen_{\Index}$ rounds.
Since $\sliceLen_{\Index}\ge\responseLen_i$, the job completes in that phase.
\end{proof}

The following theorem gives the \ResponseHeavy guarantee used by the meta-scheduler.

\begin{restatable}{theorem}{thmResponseHeavyRatio}
\label{thm:response-heavy-ratio}
For every fixed $\responseScalar>1$, the scheduler $\ALGResp$ satisfies
\begin{equation*}
    \ALGResp(\JobSet)
    =
    O_{\responseScalar}(1)\,\OptOf{\JobSet}
\end{equation*}
on every \ResponseHeavy instance $\JobSet$.
In particular, for $\responseScalar=2$,
\begin{equation*}
    \ALGResp(\JobSet)
    \le
    \frac{236}{3}\,\OptOf{\JobSet} \approx  
    78.67\,\OptOf{\JobSet}.
\end{equation*}
\end{restatable}

\xhdr{Proof overview.}
Let $\sliceLen_{-1}=0$, and let $\JobSet_{\Index}=\{i\in\JobSet:\sliceLen_{\Index-1}<\responseLen_i\le\sliceLen_{\Index}\}$ and $\NumJobsClass_{\Index}=\setsize{\JobSet_{\Index}}$ for $\Index=0,1,\ldots,\LastRespPhase$.
Write $\NumJobsGe{\Index}\triangleq\sum_{\OtherIndex=\Index}^{\LastRespPhase}\NumJobsClass_{\OtherIndex}$ for the number of jobs in classes $\Index,\Index+1,\ldots,\LastRespPhase$.
By \Cref{lem:rh:completion-cap}, every job in $\JobSet_{\Index}$ completes in phase $\Index$.
Let $\PhaseDuration_{\Index}$ denote the duration of phase $\Index$.
By the start-offset rule of $\SPS$, the last start time in phase $\Index$ is less than
$\sliceLen_{\Index}\NumJobsGe{\Index}/\paral_{\Index}$.
Thus
\begin{equation*}
    \PhaseDuration_{\Index}
    \le
    \sliceLen_{\Index}
    +
    \frac{\sliceLen_{\Index}}{\paral_{\Index}}\NumJobsGe{\Index}.
\end{equation*}
The lower bound on $\paral^*$ implies $\paral_{\Index}=\Omega(\kvmem/\sliceLen_{\Index})$ when $\sliceLen_{\Index}\le\kvmem/2$, while $\paral_{\Index}=1$ when $\sliceLen_{\Index}>\kvmem/2$.
The choice of $\paral_{\Index}$ gives $\sliceLen_{\Index}/\paral_{\Index}=O(\sliceLen_{\Index}^2/\kvmem)$, so
\begin{equation*}
    \ALGResp(\JobSet)
    =
    O\left(
        \sum_{\Index=0}^{\LastRespPhase}
        \sliceLen_{\Index}\NumJobsGe{\Index}
        +
        \frac{1}{\kvmem}
        \sum_{\Index=0}^{\LastRespPhase}
        \sliceLen_{\Index}^2\NumJobsGe{\Index}^2
    \right).
\end{equation*}
The linear sum is charged to the processing lower bound because the caps grow geometrically.
For the quadratic sum, expand the square over ordered pairs of jobs and charge each pair to the smaller response class.
Response dominance gives $\AreaOf{i}=\Theta(\responseLen_i^2)$, so this pairwise charge is bounded by the area-order lower bound.

The complete proof is deferred to \Cref{sec:response-heavy-constant}.
 \section{Meta-Scheduling with Routing}
\label{sec:meta-scheduling}

If the cases in \Cref{sec:regimes} were known in advance, the scheduler could run the corresponding sub-schedulers through a black-box meta-scheduler on one shared memory resource.
This known-partition reduction is given in \Cref{sec:meta-known-partition}. The black-box meta-scheduler runs the sub-schedulers in a round-robin order over doubling stages and preserves the total completion-time objective up to a constant factor.

The non-clairvoyant setting, however, does not directly provide this partition.
For jobs with $\initialLen_i\le\PromptCutoff$, the distinction between $\responseLen_i\le\initialLen_i$ (\PromptHeavy) and $\responseLen_i>\initialLen_i$ (\ResponseHeavy) depends on the unknown response length.
The meta-scheduler therefore needs to \textit{route} small jobs between sub-schedulers: a small job first enters the \PromptBranch; this branch either completes the job or reaches the boundary $\initialLen_i$ and certifies $\responseLen_i>\initialLen_i$.
Certified jobs are then routed to the \ResponseBranch.
The routing meta-scheduler in \Cref{sec:routing-meta-scheduler} uses the same doubling stages as the black-box meta-scheduler to combine these two events and gives a constant-competitive guarantee.

\subsection{The Black-Box Meta-Scheduler}
\label{sec:meta-known-partition}

The meta-scheduler in \Cref{alg:black-box-meta} merges a constant number of job partitions into one schedule and loses only a constant factor in the competitive guarantee.
In each fresh invocation of a sub-scheduler, every job that has already completed in the meta-scheduler is still included as a \emph{dummy job}. They follow exactly the schedule produced by the sub-scheduler, but their scheduled processing is ignored. This allows the meta-scheduler to invoke each sub-scheduler as an unchanged black box, without requiring any modification or any structural assumptions about its behavior.

\begin{algorithm2e}
\caption{Black-Box Meta-Scheduler ($\ALGBlackBoxMeta$)}
\label{alg:black-box-meta}
\KwIn{Job partition $\JobSet_1,\ldots,\JobSet_k$ and corresponding sub-schedulers $\ALG_1,\ldots,\ALG_k$.}
\For{stage $\Index=0,1,2,\ldots$}{
    \For{$\ell=1,2,\ldots,k$}{
        Run a fresh copy of $\ALG_\ell$ on $\JobSet_\ell$ for $2^\Index$ rounds\;
        \tcp{Completed jobs remain in the fresh copy as dummy jobs.}
        \tcp{Their scheduled processing is idle.}
    }
}
\end{algorithm2e}

\begin{lemma}[Black-box per-job completion control]
\label{lem:black-box-meta}
Suppose that $\JobSet$ is partitioned into disjoint sets $\JobSet_1,\ldots,\JobSet_k$.
For each $\ell=1,\ldots,k$, let $\C_i^\ell$ be the completion time of job $i\in\JobSet_\ell$ in the fresh run $\ALG_\ell(\JobSet_\ell)$.
Then the meta-scheduler $\ALGBlackBoxMeta$ in \Cref{alg:black-box-meta} completes every job $i\in\JobSet_\ell$ by time
\begin{equation*}
    \C_i^{\ALGBlackBoxMeta(\JobSet)}
    \le
    \bigl(2(k+\ell-1)+1\bigr)\C_i^\ell
    \le
    (4k-1)\C_i^\ell .
\end{equation*}
\end{lemma}

\begin{proof}
For each $i\in\JobSet_\ell$, let $\Index$ be the first stage with $2^\Index\ge \C_i^\ell$, at which job $i$ completes.
The $\ell$-th call of stage $\Index$ starts after at most another $(\ell-1)2^\Index$ rounds.
Therefore
\begin{equation*}
    \C_i^{\ALGBlackBoxMeta(\JobSet)}
    \le
    k(2^\Index-1)+(\ell-1)2^\Index+\C_i^\ell
    \le
    \bigl(2(k+\ell-1)+1\bigr)\C_i^\ell
    \le
    (4k-1)\C_i^\ell .
\end{equation*}
The claim follows.
\end{proof}

\begin{corollary}[Black-box competitive guarantee]
\label{cor:black-box-meta}
Under the assumptions and notation of \Cref{lem:black-box-meta}, suppose that each sub-scheduler $\ALG_\ell$ satisfies
\begin{equation*}
    \ALG_\ell(\JobSet_\ell)
    \le
    \approxratio_\ell \cdot \OptOf{\JobSet_\ell}.
\end{equation*}
Let $\approxratio=\max_{\ell=1,\ldots,k}\approxratio_\ell$.
Then the meta-scheduler $\ALGBlackBoxMeta$ in \Cref{alg:black-box-meta} satisfies
\begin{equation*}
    \ALGBlackBoxMeta(\JobSet)
    \le
    (4k-1)\approxratio \cdot \OptOf{\JobSet}.
\end{equation*}
\end{corollary}

\begin{proof}
By \Cref{lem:black-box-meta},
\begin{equation*}
    \ALGBlackBoxMeta(\JobSet)
    \le
    (4k-1)
    \sum_{\ell=1}^k
    \ALG_\ell(\JobSet_\ell).
\end{equation*}
The assumed sub-scheduler guarantees give
\begin{equation*}
    \ALGBlackBoxMeta(\JobSet)
    \le
    (4k-1)\approxratio
    \sum_{\ell=1}^k
    \OptOf{\JobSet_\ell}.
\end{equation*}
Since $\JobSet_1,\ldots,\JobSet_k$ form a partition of $\JobSet$,
$\sum_{\ell=1}^k\OptOf{\JobSet_\ell}\le\OptOf{\JobSet}$ holds by \Cref{lem:prelim:deletion-bound}.
\end{proof}

\subsection{The Routing Meta-Scheduler}
\label{sec:routing-meta-scheduler}

The black-box reduction assumes that every job is already assigned to the correct sub-scheduler.
The non-clairvoyant algorithm does not have this partition: for jobs with $\initialLen_i\le\PromptCutoff$, it must decide online whether a job is \PromptHeavy or \ResponseHeavy.

The \LargeBranch and \PromptBranch use the rectangle strip scheduler $\ALGRect$ defined in \Cref{sec:rectangle-strip}.
They differ only in how completion in the rectangle strip instance is interpreted in the original KV-cache instance.
\begin{itemize}
    \item \textbf{\LargeBranchCap $\LargeJobs$.} Jobs with large prompts are run by the full-width rectangle scheduler; every rectangle completion is true completion.
    \item \textbf{\PromptBranchCap $\PromptPool$.} Jobs in the prompt pool are run by the capped rectangle scheduler.
    In each run, the current pool $\PromptPool$ is represented by the capped instance $\CappedPromptPool$, with the original prompt lengths and response lengths capped at $\min\{\responseLen_i,\initialLen_i\}$.
    Completion in $\CappedPromptPool$ is either true completion in the original instance or certification that $\responseLen_i>\initialLen_i$; certified jobs move to $\ResponsePool$.
    \item \textbf{\ResponseBranchCap $\ResponsePool$.} Certified \ResponseHeavy jobs are run by the \ResponseHeavy scheduler $\ALGResp$.
\end{itemize}

$\LargeJobs$ is known in advance;
$\PromptPool$ is initialized to $\{i\in\JobSet:\initialLen_i\le\PromptCutoff\}$ and $\ResponsePool$ is initially empty.
During the algorithm, certified but unfinished jobs move from $\PromptPool$ to $\ResponsePool$.
As in the black-box meta-scheduler, completed jobs remain in their branch pool and are treated as dummy jobs in fresh runs; their scheduled processing is idle.\footnote{This is only a notational convention. The same analysis also applies if completed jobs are removed from their branch pools, because each later branch call is on a subset of the corresponding reference instance.}

\begin{algorithm2e}
\caption{Routing Meta-Scheduler ($\ALGRoute$)}
\label{alg:white-box-meta}
\KwIn{Job set $\JobSet$, memory $\kvmem$, prompt cutoff $\PromptCutoff$, and branch schedulers \LargeBranch, \PromptBranch, \ResponseBranch.}
$\LargeJobs\leftarrow\{i\in\JobSet:\initialLen_i>\PromptCutoff\}$ \tcp*{\LargeBranch}
$\PromptPool\leftarrow\{i\in\JobSet:\initialLen_i\le\PromptCutoff\}$ \tcp*{\PromptBranch}
$\ResponsePool\leftarrow\varnothing$ \tcp*{\ResponseBranch}
\For{stage $\Index=0,1,2,\ldots$}{
    Run a fresh copy of the \PromptBranch scheduler on $\PromptPool$ for $2^\Index$ rounds\;
    Move every newly certified job from $\PromptPool$ to $\ResponsePool$\;
    Run a fresh copy of the \ResponseBranch scheduler on $\ResponsePool$ for $2^\Index$ rounds\;
    Run a fresh copy of the \LargeBranch scheduler on $\LargeJobs$ for $2^\Index$ rounds\;
}
\end{algorithm2e}

\xhdr{Analysis ingredients.}
For the analysis, define
\begin{equation*}
    \SmallJobs
    \triangleq
    \JobSet\setminus\LargeJobs
    =
    \{i\in\JobSet:\initialLen_i\le\PromptCutoff\}
    ,
    \qquad
    \TruePromptPool
    \triangleq
    \{i\in\SmallJobs:\responseLen_i\le\initialLen_i\},
    \qquad
    \TrueResponsePool
    \triangleq
    \{i\in\SmallJobs:\responseLen_i>\initialLen_i\}.
\end{equation*}
Let $\CappedSmallJobs$ denote the auxiliary capped instance on $\SmallJobs$: it has the same jobs and prompt lengths as $\SmallJobs$, and response length $\min\{\responseLen_i,\initialLen_i\}$ for each job $i$.
Thus every capped prompt-pool instance $\CappedPromptPool$ that appears during the algorithm is a subinstance of $\CappedSmallJobs$.
The instance $\CappedSmallJobs$ is a relaxation of $\SmallJobs$: every feasible schedule for $\SmallJobs$ is feasible for $\CappedSmallJobs$, so $\OptOf{\CappedSmallJobs}\le\OptOf{\SmallJobs}$.

The branch calls made by the meta-scheduler are controlled by bounds defined on the corresponding full sets.
\begin{itemize}
    \item For the rectangle branches, let $\Clarge_i$ ($i\in\LargeJobs$) and $\Cprompt_i$ ($i\in\SmallJobs$) be the per-job bounds from \Cref{lem:rect-per-job} for the \LargeBranch instance on $\LargeJobs$ and the capped instance $\CappedSmallJobs$, respectively.
    The sum in \Cref{lem:rect-per-job} only decreases under job deletion, so these bounds remain valid for every rectangle branch call on a current subset.
    \item For the \ResponseBranch, let $\Cresponse_i$ be the completion time of $i$ in $\ALGResp(\TrueResponsePool)$.
    By \Cref{lem:response-heavy-subset-control}, $\Cresponse_i$ also bounds the completion time of $i$ in every \ResponseBranch call on a certified subset containing $i$.
\end{itemize}

\begin{lemma}[Routing loss]
\label{lem:routing-loss}
Consider any job $i\in\TrueResponsePool$.
Then \Cref{alg:white-box-meta} completes $i$ by time
\begin{equation*}
    \C_i^\ALGRoute
    \le
    \max\{8\Cprompt_i+\Cresponse_i,\;9\Cresponse_i\}.
\end{equation*}
\end{lemma}

\begin{proof}
Let $\rprompt_i$ be the first stage with $2^{\rprompt_i}\ge\Cprompt_i$, and let $\rresponse_i$ be the first stage with $2^{\rresponse_i}\ge\Cresponse_i$.
By minimality, $2^{\rprompt_i}<2\Cprompt_i$ and $2^{\rresponse_i}<2\Cresponse_i$.

Each stage $\OtherIndex$ consists of three calls, each of length $2^\OtherIndex$.
Hence the total time before stage $\Index$ is
\begin{equation*}
    3\sum_{\OtherIndex=0}^{\Index-1}2^\OtherIndex
    =
    3(2^\Index-1).
\end{equation*}
The \PromptBranch call in stage $\Index$ ends by time $3(2^\Index-1)+2^\Index$, and the \ResponseBranch call in the same stage starts at that time.

If $\rprompt_i\le\rresponse_i$, then job $i$ has been certified by the time the \ResponseBranch call of stage $\rresponse_i$ starts.
By \Cref{lem:response-heavy-subset-control}, this call completes $i$ within $\Cresponse_i$ additional rounds.
Therefore
\begin{equation*}
    \C_i^\ALGRoute
    \le
    3(2^{\rresponse_i}-1)+2^{\rresponse_i}+\Cresponse_i
    \le
    9\Cresponse_i .
\end{equation*}
If $\rprompt_i>\rresponse_i$, then job $i$ is certified during the \PromptBranch call of stage $\rprompt_i$ and is included in the following \ResponseBranch call of the same stage.
Since $\rprompt_i>\rresponse_i$, this \ResponseBranch call has length at least $\Cresponse_i$.
Again by \Cref{lem:response-heavy-subset-control},
\begin{equation*}
    \C_i^\ALGRoute
    \le
    3(2^{\rprompt_i}-1)+2^{\rprompt_i}+\Cresponse_i
    \le
    8\Cprompt_i+\Cresponse_i .
\end{equation*}
Combining the two cases proves the lemma.
\end{proof}

\thmMain*

\begin{proof}
Set $\PromptCutoff=\kvmem/4$ and $\responseScalar=2$.
Run the routing meta-scheduler $\ALGRoute$ (\Cref{alg:white-box-meta}).
We first bound each job's completion time by the corresponding branch cost, and then compare the branch sums with $\OPT$.

By \Cref{lem:black-box-meta}, applied to the \PromptBranch and \LargeBranch calls,
\begin{equation*}
    \C_i^\ALGRoute\le7\Cprompt_i
    \quad (i\in\TruePromptPool),
    \qquad
    \C_i^\ALGRoute\le11\Clarge_i
    \quad (i\in\LargeJobs).
\end{equation*}
Together with \Cref{lem:routing-loss}, this gives
\begin{equation*}
    \ALGRoute(\JobSet)
    \le
    11\sum_{i\in\LargeJobs}\Clarge_i
    +
    8\sum_{i\in\SmallJobs}\Cprompt_i
    +
    9\sum_{i\in\TrueResponsePool}\Cresponse_i .
\end{equation*}
With $\PromptCutoff=\kvmem/4$, \Cref{prop:rectangle-large-guarantee,prop:rectangle-prompt-heavy-guarantee} give
\begin{equation*}
    \sum_{i\in\LargeJobs}\Clarge_i
    \le
    36\,\OptOf{\LargeJobs},
    \qquad
    \sum_{i\in\SmallJobs}\Cprompt_i
    \le
    36\,\OptOf{\CappedSmallJobs}
    \le
    36\,\OptOf{\SmallJobs}.
\end{equation*}
For the response branch, \Cref{thm:response-heavy-ratio} and the deletion bound (\Cref{lem:prelim:deletion-bound}) give
\begin{equation*}
    \sum_{i\in\TrueResponsePool}\Cresponse_i
    =
    \ALGResp(\TrueResponsePool)
    \le
    \frac{236}{3}\,\OptOf{\TrueResponsePool}
    \le
    \frac{236}{3}\,\OptOf{\SmallJobs}.
\end{equation*}
Since $\LargeJobs$ and $\SmallJobs$ partition $\JobSet$, combining with the deletion bound (\Cref{lem:prelim:deletion-bound}) gives
\begin{equation*}
    \ALGRoute(\JobSet)
    \le
    396\,\OptOf{\LargeJobs}
    +
    (288+708)\,\OptOf{\SmallJobs}
    \le
    996\,\OPT .
\end{equation*}
\end{proof}

\begin{remark}[Computational complexity of \ALGRoute]
Although the scheduling model describes decisions round by round,
the total number of discrete events across \ALGRect, \ALGResp, and \ALGRoute
is bounded by $\mathrm{poly}(\NumberJobs, \log \kvmem)$.
An event-driven implementation that jumps between successive events
therefore runs in polynomial time.
\end{remark} \section{Conclusion and Discussion}
\label{sec:conclusion}

We have presented the first constant-competitive non-clairvoyant scheduling algorithm for batched LLM inference under a hard KV-cache memory budget with arbitrary prompt lengths and arbitrary response lengths. The core algorithmic idea is a regime-aware routing framework that decomposes jobs into three geometric regimes---large jobs, small prompt-heavy jobs, and small response-heavy jobs---and handles each with a specialized sub-scheduler. A meta-scheduler time-shares the memory budget across these regimes and dynamically routes jobs as their execution reveals their true geometric class. This yields an $O(1)$-competitive guarantee for total completion time, and the same framework extends to other variants and extensions such as makespan minimization and online arrivals.

From a systems perspective, our results provide a principled foundation for a design principle already present in modern serving stacks: global request routing across replicas. Rather than introducing routing complexity inside a single replica, our framework shows that partitioning jobs by their geometric structure---large jobs, prompt-heavy jobs, and response-heavy jobs---and dispatching each class to a dedicated replica with a simple specialized local scheduler can simultaneously achieve worst-case robustness and practical efficiency. Because the theoretical meta-scheduler must time-share a single memory budget across heterogeneous classes, a multi-replica implementation actually attains a better competitive ratio than our single-replica bound, while remaining relatively simpler to deploy with existing continuous-batching and block-based cache managers. In this sense, simple local rules, properly routed by geometric class, suffice for both theoretical and practical performance.

\bibliographystyle{plainnat}
\bibliography{refs}

\appendix
\section{Clairvoyant Scheduling Guarantees}
\label{sec:clairvoyant}

The clairvoyant benchmark has three useful forms.
Area-order greedy is nearly optimal when every job is small compared with the memory budget.
For general instances, splitting jobs by full rectangle width and combining the two resulting schedules gives a constant approximation.
A direct rectangle packing argument gives a smaller constant by using a known strip packing guarantee as a black box.
The schedules below reserve the full rectangle width of each active job:
\begin{equation*}
    \FullWidthOf{i}
    \triangleq
    \initialLen_i+\responseLen_i,
\end{equation*}
i.e., reserving the full rectangle of each job.

\subsection{Large-Memory Case}
\label{sec:clairvoyant-large-memory}
The large-memory case is a KV-cache analogue of the small-demand WSVF guarantee for capacitated machines~\cite[Theorem~1]{CCZ24}.

\begin{algorithm2e}[H]
\caption{Clairvoyant Area-Order Greedy (\Greedy)}
\label{alg:clairvoyant-area-order}
\KwIn{Job set $\JobSet$ and memory budget $\kvmem$.}
Order the jobs in $\JobSet$ as $\pi(1),\ldots,\pi(\setsize{\JobSet})$ so that $\AreaOf{\pi(1)}\le\cdots\le\AreaOf{\pi(\setsize{\JobSet})}$\;
$\ActiveBatch\leftarrow\varnothing$ and $q\leftarrow 1$\;
\For{round $t=0,1,2,\ldots$ until ($q>\setsize{\JobSet}$ and $\ActiveBatch=\varnothing$)}{
    \While{$q\le\setsize{\JobSet}$ and
    $\sum_{i\in\ActiveBatch}\FullWidthOf{i}
    +\FullWidthOf{\pi(q)}\le\kvmem$}{
        Start job $\pi(q)$ and add it to $\ActiveBatch$\;
        $q\leftarrow q+1$\;
    }
    Process every job in $\ActiveBatch$ for one round, and remove completed jobs from $\ActiveBatch$\;
}
\end{algorithm2e}

\begin{theorem}[Bounded-width area-order greedy]
\label{thm:clairvoyant-bounded-width}
Fix $\varepsilon\in(0,1)$.
Suppose that every job in $\JobSet$ satisfies
\begin{equation*}
    \FullWidthOf{i}
    \le
    \varepsilon\kvmem .
\end{equation*}
Then the schedule produced by $\Greedy$ satisfies
\begin{equation*}
    \Greedy(\JobSet)
    \le
    \left(1+\frac{2}{1-\varepsilon}\right)\,\OptOf{\JobSet}.
\end{equation*}
\end{theorem}

\begin{proof}
Let $\StartTimeOf{\pi(r)}$ be the start time of job $\pi(r)$.
Before time $\StartTimeOf{\pi(r)}$, if job $\pi(r)$ has not yet been admitted, then the next unstarted job has reserved width at most $\varepsilon\kvmem$.
Thus, after the greedy admission step in each such round, the active reserved width is larger than $(1-\varepsilon)\kvmem$.
All active jobs before time $\StartTimeOf{\pi(r)}$ belong to $\{\pi(1),\ldots,\pi(r-1)\}$.
Therefore,
\begin{equation*}
    \StartTimeOf{\pi(r)}
    \le
    \frac{1}{(1-\varepsilon)\kvmem}
    \sum_{q=1}^{r-1}
    \FullWidthOf{\pi(q)}\responseLen_{\pi(q)} .
\end{equation*}
Since job $\pi(r)$ completes $\responseLen_{\pi(r)}$ rounds after it starts,
\begin{equation*}
    \C_{\pi(r)}
    \le
    \responseLen_{\pi(r)}
    +
    \frac{1}{(1-\varepsilon)\kvmem}
    \sum_{q=1}^{r-1}
    \FullWidthOf{\pi(q)}\responseLen_{\pi(q)} .
\end{equation*}
Summing over all ranks gives
\begin{equation*}
    \Greedy(\JobSet)
    \le
    \sum_{i\in\JobSet}\responseLen_i
    +
    \frac{1}{(1-\varepsilon)\kvmem}
    \sum_{q=1}^{\NumberJobs}
    (\NumberJobs-q)
    \FullWidthOf{\pi(q)}\responseLen_{\pi(q)} .
\end{equation*}
For every job $i$,
\begin{equation*}
    \FullWidthOf{i}\responseLen_i
    =
    \initialLen_i\responseLen_i+\responseLen_i^2
    \le
    2\AreaOf{i}.
\end{equation*}
Hence
\begin{equation*}
    \Greedy(\JobSet)
    \le
    \sum_{i\in\JobSet}\responseLen_i
    +
    \frac{2}{(1-\varepsilon)\kvmem}
    \sum_{q=1}^{\NumberJobs}
    (\NumberJobs-q)\AreaOf{\pi(q)} .
\end{equation*}
Since $\pi$ is the nondecreasing area order, \Cref{lem:prelim:processing-lower-bound,lem:prelim:area-lower-bound} give
\begin{equation*}
    \Greedy(\JobSet)
    \le
    \left(1+\frac{2}{1-\varepsilon}\right)\,\OptOf{\JobSet}.
\end{equation*}
\end{proof}

\begin{corollary}[Large-memory area-order greedy]
\label{cor:clairvoyant-large-memory}
If
\(
    \max_{i\in\JobSet}
    {\FullWidthOf{i}}
    =
    o(\kvmem),
\)
then the clairvoyant area-order schedule in \Cref{thm:clairvoyant-bounded-width} is a $(3+o(1))$-approximation.
Equivalently,
\begin{equation*}
    \Greedy(\JobSet)
    \le
    (3+o(1))\,\OPT(\JobSet) .
\end{equation*}
\end{corollary}

\subsection{General Case via Black-Box Meta-Scheduling}
\label{sec:clairvoyant-general}

For a general clairvoyant instance, view each job as a rectangle with width $\FullWidthOf{i}$ and height $\responseLen_i$.
The jobs are separated by rectangle width.
The small-width part uses \Greedy with the true full widths.
The large-width part uses the same \Greedy routine on an auxiliary instance in which every job reserves rectangle width $\kvmem$.
The two schedules are merged by the black-box meta-scheduler of \Cref{sec:meta-known-partition}.

\begin{algorithm2e}[H]
\caption{Clairvoyant Width-Splitting Scheduler (\ClairvoyantMix)}
\label{alg:clairvoyant-width-mix}
\KwIn{Job set $\JobSet$, memory $\kvmem$, and threshold $\WidthThreshold\in(0,1)$.}
Set $\JobSet_{\le\WidthThreshold}\leftarrow\{i\in\JobSet:\FullWidthOf{i}\le\WidthThreshold\kvmem\}$\;
Set $\JobSet_{>\WidthThreshold}\leftarrow\JobSet\setminus\JobSet_{\le\WidthThreshold}$\;
Let $\JobSet_{>\WidthThreshold}^{\kvmem}$ be $\JobSet_{>\WidthThreshold}$ with reserved rectangle width $\kvmem$ for each job\;
Run \Cref{alg:black-box-meta} with sub-schedulers $\Greedy(\JobSet_{\le\WidthThreshold})$ and $\Greedy(\JobSet_{>\WidthThreshold}^{\kvmem})$\;
\end{algorithm2e}

\begin{lemma}[Full-width auxiliary greedy call]
\label{lem:clairvoyant-full-width}
For $\WidthThreshold=1/2$, the auxiliary instance $\JobSet_{>\WidthThreshold}^{\kvmem}$ defined in \Cref{alg:clairvoyant-width-mix} satisfies
\begin{equation*}
    \Greedy(\JobSet_{>\WidthThreshold}^{\kvmem})
    \le
    4\,\OptOf{\JobSet_{>\WidthThreshold}}.
\end{equation*}
\end{lemma}

\begin{proof}
Set $m\triangleq\setsize{\JobSet_{>\WidthThreshold}}$.
Let $\pi(1),\ldots,\pi(m)$ be the nondecreasing true area order of $\JobSet_{>\WidthThreshold}$.
In the auxiliary instance, every job reserves $\kvmem$ and has area key $\kvmem\responseLen_i$.
Thus \Greedy runs one job at a time in nondecreasing $\responseLen_i$.
This order has no larger total completion time than the sequential schedule that follows the order $\pi$, so
\begin{equation*}
    \Greedy(\JobSet_{>\WidthThreshold}^{\kvmem})
    \le
    \sum_{r=1}^{m}
    (m-r+1)\responseLen_{\pi(r)} .
\end{equation*}
For every $i\in\JobSet_{>\WidthThreshold}$,
\begin{equation*}
    \AreaOf{i}
    \ge
    \frac{\FullWidthOf{i}\responseLen_i}{2}
    >
    \frac{\kvmem\responseLen_i}{4}.
\end{equation*}
It follows that
\begin{equation*}
    \Greedy(\JobSet_{>\WidthThreshold}^{\kvmem})
    \le
    \frac{4}{\kvmem}
    \sum_{r=1}^{m}
    (m-r+1)\AreaOf{\pi(r)}
    \le
    4\,\OptOf{\JobSet_{>\WidthThreshold}},
\end{equation*}
where the last inequality follows from \Cref{lem:prelim:area-lower-bound} applied to $\JobSet_{>\WidthThreshold}$.
\end{proof}

\begin{theorem}
\label{thm:clairvoyant-width-mix}
The schedule produced by \ClairvoyantMix satisfies
\begin{equation*}
    \ClairvoyantMix(\JobSet)
    \le
    28\OPT (\JobSet).
\end{equation*}
\end{theorem}

\begin{proof}
For $\WidthThreshold=1/2$, the same proof as \Cref{thm:clairvoyant-bounded-width} gives
\begin{equation*}
    \Greedy(\JobSet_{\le\WidthThreshold})
    \le
    5\,\OptOf{\JobSet_{\le\WidthThreshold}},
\end{equation*}
and \Cref{lem:clairvoyant-full-width} gives
\begin{equation*}
    \Greedy(\JobSet_{>\WidthThreshold}^{\kvmem})
    \le
    4\,\OptOf{\JobSet_{>\WidthThreshold}}.
\end{equation*}
The branch-dependent bound in \Cref{lem:black-box-meta} gives a factor of $5$ for the first branch and a factor of $7$ for the second branch.
Since \Cref{alg:clairvoyant-width-mix} places the small-width branch first, it gives
\begin{equation*}
    \ClairvoyantMix(\JobSet)
    \le
    5\Greedy(\JobSet_{\le\WidthThreshold})
    +
    7\Greedy(\JobSet_{>\WidthThreshold}^{\kvmem}) .
\end{equation*}
Together with \Cref{lem:prelim:deletion-bound}, this implies
\begin{equation*}
    \ClairvoyantMix(\JobSet)
    \le
    25\,\OptOf{\JobSet_{\le\WidthThreshold}}
    +
    28\,\OptOf{\JobSet_{>\WidthThreshold}}
    \le
    28\OPT (\JobSet).
\end{equation*}
This proves the theorem.
\end{proof}

\subsection{General Case with Direct Rectangle-Packing}
\label{sec:clairvoyant-rectangle-packing}

This subsection gives a second clairvoyant bound using a stronger geometric primitive.
The previous two subsections use area-order scheduling directly;
here we call a strip packing algorithm as a black box, which gives a stronger makespan guarantee for each stage.
The conversion from stage makespan guarantees to total completion time follows the doubling framework of the single-machine capacitated scheduling algorithm of \cite[Theorem~7.1]{CCZ24}.
Two simplifications apply in the present unweighted rectangle instance:
\begin{itemize}
    \item The knapsack selection step for weighted rectangles becomes the exact prefix after sorting by rectangle area.
    \item For rectangles with width at most $\kvmem$, height at most $L$, and total area at most $\kvmem L$,
    the $2$-approximation algorithm by \cite{Steinberg97} packs them in height at most $2L$.
\end{itemize}
These two changes give an $8$-approximation for the rectangle instance, improving the $12+\varepsilon$ factor from the corresponding framework of \cite{CCZ24}.
The conversion from KV-cache jobs to rectangles loses another factor of $2$ against the area lower bound, yielding the $16$-approximation below.

\begin{algorithm2e}
\caption{Clairvoyant Rectangle-Packing Scheduler (\PackAndSchedule)}
\label{alg:clairvoyant-rectangle-packing}
\KwIn{Job set $\JobSet$ and memory $\kvmem$.}
\For{stage $\Index=0,1,2,\ldots$ until all jobs have completed}{
    Let $\EligibleJobs_{\Index}\leftarrow\{i\in\JobSet:\responseLen_i\le 2^\Index\}$, ordered by nondecreasing $\FullWidthOf{i}\responseLen_i$\;
    Let $\AlgCompletedJobs_{\Index}$ be the longest prefix of $\EligibleJobs_{\Index}$ whose total rectangle area is at most $\kvmem 2^\Index$\;
    Pack the rectangles $\{(\FullWidthOf{i},\responseLen_i):i\in\AlgCompletedJobs_{\Index}\}$ in a $\kvmem$-width, $2^{\Index+1}$-height strip by Steinberg's algorithm \citep{Steinberg97}\;
    Run the packed schedule \tcp*{Completed jobs are treated as dummy jobs.}
}
\end{algorithm2e}

\begin{theorem}
\label{thm:clairvoyant-rectangle-packing}
The schedule produced by \PackAndSchedule satisfies
\begin{equation*}
    \PackAndSchedule(\JobSet)
    \le
    16\OPT(\JobSet).
\end{equation*}
\end{theorem}

\begin{proof}
After stage $\Index$, the algorithm has completed at least $\setsize{\AlgCompletedJobs_{\Index}}$ jobs by time
\begin{equation*}
    2\sum_{\OtherIndex=0}^{\Index}2^\OtherIndex
    <
    4\cdot 2^\Index .
\end{equation*}
The optimal schedule completes at most $\setsize{\AlgCompletedJobs_{\Index}}$ jobs by time $2^{\Index-1}$.
To see that, let $\OptCompletedJobs_{\Index}$ be the jobs completed by an optimal schedule by that time.
Every job in $\OptCompletedJobs_{\Index}$ has $\responseLen_i\le 2^\Index$, and these jobs have total rectangle area at most
\begin{equation*}
    \sum_{i\in\OptCompletedJobs_{\Index}} \FullWidthOf{i}\responseLen_i
    \le
    2\sum_{i\in\OptCompletedJobs_{\Index}}\AreaOf{i}
    \le
    \kvmem 2^\Index .
\end{equation*}
By the maximality of the prefix $\AlgCompletedJobs_{\Index}$, this implies $\setsize{\OptCompletedJobs_{\Index}}\le\setsize{\AlgCompletedJobs_{\Index}}$.

For ranks $\setsize{\AlgCompletedJobs_{\Index-1}}+1,\ldots,\setsize{\AlgCompletedJobs_{\Index}}$ with $\Index\ge0$ (using the convention $\AlgCompletedJobs_{-1}=\varnothing$), the algorithm completes these ranks by time less than $4\cdot2^\Index$, while the optimal schedule cannot complete these ranks before time $2^{\Index-2}$.
Summing over all ranks gives $\PackAndSchedule(\JobSet)\le16\OPT$.
\end{proof}

\begin{remark}
The same argument improves the single-machine weighted capacitated scheduling bound of \cite[Theorem~7.1]{CCZ24} from $12+\varepsilon$ to $8+\varepsilon$.
It keeps their resource-augmented knapsack step, adds the natural scale eligibility condition on processing times ($\EligibleJobs_{\Index}$ in \Cref{alg:clairvoyant-rectangle-packing}), and replaces the factor-$3$ packing step based on \cite{JZ07} by Steinberg's factor-$2$ strip packing theorem.
\end{remark}
 \section{Complete Non-Clairvoyant Rectangle Strip Scheduling}
\label{sec:rectangle-strip-complete}

The rectangle strip guarantee in \Cref{thm:rectangle-strip} assumes a useful width range.
This section removes that assumption by mixing two rectangle strip schedulers through the black-box meta-scheduler.

\begin{algorithm2e}
\caption{Complete Rectangle Strip Scheduler ($\ALGRectMix$)}
\label{alg:rectangle-strip-complete}
\KwIn{Rectangle job set $\JobSet$, widths $\{\rectWidth_i\}_{i\in\JobSet}$, memory $\kvmem$, and threshold $\WidthThreshold\in(0,1)$.}
Set $\JobSet_{\le\WidthThreshold}\leftarrow\{i\in\JobSet:\rectWidth_i\le\WidthThreshold\kvmem\}$\;
Set $\JobSet_{>\WidthThreshold}\leftarrow\{i\in\JobSet:\rectWidth_i>\WidthThreshold\kvmem\}$\;
Run the black-box meta-scheduler with $\ALGRect$ on $\JobSet_{\le\WidthThreshold}$ and $\ALGRect$ on $\JobSet_{>\WidthThreshold}$\;
\end{algorithm2e}

\begin{theorem}
\label{thm:rectangle-strip-complete}
For every rectangle strip instance with $\rectWidth_i\le\kvmem$ for all $i\in\JobSet$, there is a non-clairvoyant scheduling algorithm whose total completion time is within a constant factor of the optimal clairvoyant total completion time:
\begin{equation*}
    \ALGRectMix(\JobSet)
    =
    O(1) \cdot \OptRectOf{\JobSet}.
\end{equation*}
\end{theorem}

\begin{proof}
Run \Cref{alg:rectangle-strip-complete} with $\WidthThreshold=(\sqrt{149}-11)/2$.
By \Cref{thm:rectangle-strip}, the two branch ratios are
\begin{equation*}
    4+\frac{8}{1-\WidthThreshold}
    \qquad\text{and}\qquad
    4+\frac{8}{\WidthThreshold}.
\end{equation*}
By \Cref{lem:black-box-meta}, the first branch loses a factor of $5$ and the second branch loses a factor of $7$.
Hence
\begin{equation*}
    \ALGRectMix(\JobSet)
    \le
    5\left(4+\frac{8}{1-\WidthThreshold}\right)
    \OptRectOf{\JobSet_{\le\WidthThreshold}}
    +
    7\left(4+\frac{8}{\WidthThreshold}\right)
    \OptRectOf{\JobSet_{>\WidthThreshold}}.
\end{equation*}
The chosen $\WidthThreshold$ makes both coefficients equal to $72+4\sqrt{149}$.
Since $\JobSet_{\le\WidthThreshold}$ and $\JobSet_{>\WidthThreshold}$ partition $\JobSet$, \Cref{lem:prelim:deletion-bound} gives the theorem.
\end{proof}
 \section{Competitive Ratio for the \ResponseHeavyCap Scheduler}
\label{sec:response-heavy-constant}

This appendix proves the explicit constant used in \Cref{thm:response-heavy-ratio}.
For the rest of the appendix, fix $\responseScalar=2$ and a \ResponseHeavy job set $\JobSet$;
for every fixed $\responseScalar>1$, the same argument gives the stated $O_{\responseScalar}(1)$ bound.
Let $\LastRespPhase\triangleq\ceil{\log_2\kvmem}$.
Let $\sliceLen_{-1}=0$ and define the phase classes
\begin{equation*}
    \JobSet_{\Index}
    =
    \{i\in\JobSet:\sliceLen_{\Index-1}<\responseLen_i\le\sliceLen_{\Index}\},
    \qquad
    \NumJobsClass_{\Index}=\setsize{\JobSet_{\Index}},
    \qquad
    \NumJobsGe{\Index}=\sum_{\OtherIndex=\Index}^{\LastRespPhase}\NumJobsClass_{\OtherIndex},
    \qquad
    \NumJobsGt{\Index}=\sum_{\OtherIndex=\Index+1}^{\LastRespPhase}\NumJobsClass_{\OtherIndex}
\end{equation*}
for $\Index=0,1,\ldots,\LastRespPhase$.
Let $\lowerSlice_0\triangleq1$ and $\lowerSlice_{\Index}\triangleq\sliceLen_{\Index-1}$ for $\Index\ge1$.
Then every job $i\in\JobSet_{\Index}$ satisfies
\begin{equation*}
    \responseLen_i\ge\lowerSlice_{\Index},
    \qquad
    \sliceLen_{\Index}\le2\lowerSlice_{\Index} .
\end{equation*}

\begin{lemma}[Phase efficiency]
\label{lem:resp-constant-phase-efficiency}
For every phase $\Index$,
\begin{equation*}
    \frac{\sliceLen_{\Index}}{\paral_{\Index}}
    \le
    \frac{7}{2}\cdot\frac{\sliceLen_{\Index}^2}{\kvmem}.
\end{equation*}
\end{lemma}

\begin{proof}
Write $\sliceLen=\sliceLen_{\Index}$ and $\paral=\paral_{\Index}$.
If $\sliceLen>\kvmem/2$, then $\paral=1$, and the claim follows from $\kvmem<2\sliceLen$.

Assume $\sliceLen\le\kvmem/2$.
Then the proxy prompt is $\ProxyPrompt_{\Index}=\sliceLen$.
If $\paral=1$, the maximality of $\paral$ implies that $\paral=2$ is infeasible.
Thus
\begin{equation*}
    \MemPeak(2,\sliceLen,\sliceLen)
    =
    2\sliceLen
    +
    \frac{3\sliceLen+2-\gcd(\sliceLen,2)}2
    >
    \kvmem .
\end{equation*}
If $\sliceLen$ is even, this gives $\kvmem<7\sliceLen/2$.
If $\sliceLen$ is odd, then $\MemPeak(2,\sliceLen,\sliceLen)=(7\sliceLen+1)/2$ is an integer, so $\kvmem\le(7\sliceLen-1)/2<7\sliceLen/2$.
In both cases, $\sliceLen\le(7/2)\sliceLen^2/\kvmem$.

It remains to consider $\paral\ge2$.
By maximality, $\paral+1$ is infeasible.
Using $\gcd(\sliceLen,\paral+1)\ge1$,
\[
    \kvmem
    <
    \MemPeak(\paral+1,\sliceLen,\sliceLen)
    \le
    \sliceLen(\paral+1)
    +
    \frac{\sliceLen(\paral+1)+\sliceLen+\paral}{2}
    =
    \frac{3\paral\sliceLen+4\sliceLen+\paral}{2}
    \le
    \frac{7\paral\sliceLen}{2}.
\]
The last inequality uses $\paral\ge2$ and $\sliceLen\ge1$.
Rearranging gives the claim.
\end{proof}

Define the class-area lower-bound quantity
\begin{equation*}
    \ClassAreaLB
    \triangleq
    \sum_{\Index=0}^{\LastRespPhase}
    \lowerSlice_{\Index}^2
    \left(
        \frac{\NumJobsClass_{\Index}(\NumJobsClass_{\Index}+1)}2
        +
        \NumJobsClass_{\Index}\NumJobsGt{\Index}
    \right).
\end{equation*}

\begin{lemma}[Class-area lower bound]
\label{lem:resp-constant-class-area}
The optimum satisfies
\begin{equation*}
    \OptOf{\JobSet}
    \ge
    \frac{\ClassAreaLB}{2\kvmem}.
\end{equation*}
\end{lemma}

\begin{proof}
For every class $\Index$ and every job $i\in\JobSet_{\Index}$, response dominance gives
\begin{equation*}
    \AreaOf{i}
    \ge
    \frac{\responseLen_i^2}{2}
    \ge
    \frac{\lowerSlice_{\Index}^2}{2}.
\end{equation*}
By \Cref{lem:prelim:area-lower-bound}, the lower-bound expression is
\[
    \frac1\kvmem
    \sum_{\ell=1}^{\setsize{\JobSet}}
    \left(\setsize{\JobSet}-\ell+1\right)
    \AreaOf{\pi(\ell)},
\]
where $\pi$ orders jobs by nondecreasing area.
This expression is the sum of all job areas plus, for every unordered pair of jobs, the smaller area in that pair, divided by $\kvmem$.
Its diagonal and within-class pairs from class $\Index$ contribute at least
$\lowerSlice_{\Index}^2\NumJobsClass_{\Index}(\NumJobsClass_{\Index}+1)/(4\kvmem)$.
The pairs with one job in class $\Index$ and one job in a later class contribute at least
$\lowerSlice_{\Index}^2\NumJobsClass_{\Index}\NumJobsGt{\Index}/(2\kvmem)$.
Summing these contributions over all classes gives $\ClassAreaLB/(2\kvmem)$.
\end{proof}

\thmResponseHeavyRatio*

\begin{proof}
Let $\PhaseDuration_{\Index}$ be the actual elapsed time of the $\SPS$ call in phase $\Index$.
By the start-offset rule of $\SPS$, the last start time in an $\SPS$ call on at most $\NumJobsGe{\Index}$ jobs is less than
$\sliceLen_{\Index}\NumJobsGe{\Index}/\paral_{\Index}$, so
\begin{equation*}
    \PhaseDuration_{\Index}
    \le
    \sliceLen_{\Index}
    +
    \frac{\sliceLen_{\Index}}{\paral_{\Index}}\NumJobsGe{\Index}.
\end{equation*}
By \Cref{lem:rh:completion-cap}, every job in classes $\Index,\Index+1,\ldots,\LastRespPhase$ is unfinished before phase $\Index$, so
\[
    \ALGResp(\JobSet)
    \le
    \sum_{\Index=0}^{\LastRespPhase}\PhaseDuration_{\Index}\NumJobsGe{\Index}
    \le
    \sum_{\Index=0}^{\LastRespPhase}\sliceLen_{\Index}\NumJobsGe{\Index}
    +
    \sum_{\Index=0}^{\LastRespPhase}
    \frac{\sliceLen_{\Index}}{\paral_{\Index}}
    \NumJobsGe{\Index}^2 .
\]

For the linear term, geometric growth gives $\sum_{\Index=0}^\ell\sliceLen_{\Index}\le4\lowerSlice_\ell$ for every class $\ell$.
Therefore
\[
    \sum_{\Index=0}^{\LastRespPhase}\sliceLen_{\Index}\NumJobsGe{\Index}
    =
    \sum_{\ell=0}^{\LastRespPhase}
    \NumJobsClass_\ell\sum_{\Index=0}^\ell\sliceLen_{\Index}
    \le
    4\sum_{\ell=0}^{\LastRespPhase}\NumJobsClass_\ell\lowerSlice_\ell
    \le
    4\sum_{i\in\JobSet}\responseLen_i
    \le
    4\,\OptOf{\JobSet}.
\]

Set $c=7/2$.
By \Cref{lem:resp-constant-phase-efficiency},
\begin{equation*}
    \sum_{\Index=0}^{\LastRespPhase}
    \frac{\sliceLen_{\Index}}{\paral_{\Index}}
    \NumJobsGe{\Index}^2
    \le
    \frac{c}{\kvmem}
    \sum_{\Index=0}^{\LastRespPhase}
    \sliceLen_{\Index}^2
    \NumJobsGe{\Index}^2 .
\end{equation*}
Expanding over ordered pairs of jobs and charging each pair to the smaller response class gives
\[
    \sum_{\Index=0}^{\LastRespPhase}\sliceLen_{\Index}^2\NumJobsGe{\Index}^2
    \le
    \frac{16}{3}
    \sum_{\ell,s=0}^{\LastRespPhase}
    \NumJobsClass_\ell\NumJobsClass_s
    \lowerSlice_{\min\{\ell,s\}}^2
    \le
    \frac{32}{3}\ClassAreaLB.
\]
The first inequality uses $\sum_{\Index=0}^\ell\sliceLen_{\Index}^2\le(16/3)\lowerSlice_\ell^2$ for every class $\ell$.
The second inequality compares ordered pairs with the definition of $\ClassAreaLB$:
within a class, $\NumJobsClass_\ell^2\le \NumJobsClass_\ell(\NumJobsClass_\ell+1)$, and between two different classes the ordered-pair count is twice the unordered-pair count.

Combining the linear and quadratic bounds gives
\begin{equation*}
    \ALGResp(\JobSet)
    \le
    4\,\OptOf{\JobSet}
    +
    \frac{112}{3\kvmem}\ClassAreaLB.
\end{equation*}
By \Cref{lem:resp-constant-class-area}, $\ClassAreaLB\le2\kvmem\,\OptOf{\JobSet}$.
Hence
\begin{equation*}
    \ALGResp(\JobSet)
    \le
    \left(4+\frac{224}{3}\right)\,\OptOf{\JobSet}
    =
    \frac{236}{3}\,\OptOf{\JobSet}
    \approx
    78.67\,\OptOf{\JobSet}.
\end{equation*}
\end{proof}

\begin{lemma}[Response subset control]
\label{lem:response-heavy-subset-control}
Fix a \ResponseHeavy set $\TrueResponsePool$ and a deterministic order used by $\ALGResp$.
For each $i\in\TrueResponsePool$, let $\Cresponse_i$ be the completion time of $i$ in $\ALGResp(\TrueResponsePool)$.
Then, for every $\SubJobSet\subseteq\TrueResponsePool$ and every $i\in\SubJobSet$,
job $i$ completes in $\ALGResp(\SubJobSet)$ by time $\Cresponse_i$.
\end{lemma}

\begin{proof}
Fix any job $i\in\SubJobSet$, and let $\Index$ be the first phase with $\sliceLen_{\Index}\ge\responseLen_i$.
Run $\ALGResp(\SubJobSet)$ using the order induced by the fixed order on $\TrueResponsePool$.
By \Cref{lem:rh:completion-cap}, job $i$ completes in phase $\Index$ in both $\ALGResp(\SubJobSet)$ and $\ALGResp(\TrueResponsePool)$.
For every phase $\OtherIndex<\Index$, whether a job is considered in that phase depends only on its prompt length, response length, and the cap sequence.
Thus the jobs considered in the run on $\SubJobSet$ form a subsequence of the jobs considered in the run on $\TrueResponsePool$.
For every remaining job in this subsequence, its rank is no larger than its rank in the run on $\TrueResponsePool$.
Since the $\SPS$ start offset is $\floor{j\sliceLen_{\OtherIndex}/\paral_{\OtherIndex}}$ for rank $j$, each remaining job starts no later in the run on $\SubJobSet$ than in the run on $\TrueResponsePool$.
The real running time of every remaining job in phase $\OtherIndex$ is unchanged, so the duration of phase $\OtherIndex$ in the run on $\SubJobSet$ is no larger than the duration of phase $\OtherIndex$ in the run on $\TrueResponsePool$.
The same subsequence property holds inside phase $\Index$:
the rank of job $i$ in the run on $\SubJobSet$ is at most its rank in the run on $\TrueResponsePool$, so the $\SPS$ start offset of job $i$ is no later in the run on $\SubJobSet$ than in the run on $\TrueResponsePool$.
After job $i$ starts in phase $\Index$, the remaining time until completion is $\responseLen_i$ in both runs.
Combining these components gives the claim.
\end{proof}
 \section{Alternative Objectives}
\label{sec:alternative-objectives}

This appendix gives two extensions of the meta-scheduling framework.
\begin{itemize}
    \item \Cref{sec:alternative-makespan} shows that the same algorithm is also constant-competitive for makespan.
    \item \Cref{sec:alternative-online-completion-time} extends the routing framework to total completion time with online arrivals.
\end{itemize}
For online arrivals, the natural latency objective is total flow time.
Under (oblivious) adversarial arrivals, however, online KV-cache scheduling admits an $\Omega(\sqrt{\NumberJobs})$ lower bound on the competitive ratio for total flow time, even in the identical-prompt special case~\citep{JJMMPZ25,FYZ26}.
Total completion time is the closest feasible alternative:
for a fixed arrival instance, the two objectives induce the same optimal schedules.
Online completion-time objectives have also been studied in classical scheduling~\citep{ES03,SP05}.

\subsection{Makespan}
\label{sec:alternative-makespan}

For a job set $\SubJobSet\subseteq\JobSet$, let $\OptMSpanOf{\SubJobSet}$ be the optimal clairvoyant makespan for scheduling exactly the jobs in $\SubJobSet$.
For a schedule produced by an algorithm $\ALG$, let $\MSpan{\ALG(\SubJobSet)}$ denote its makespan.

\begin{lemma}[Makespan lower bounds]
\label{lem:alternative:makespan-lower-bounds}
For every $\SubJobSet\subseteq\JobSet$,
\begin{equation*}
    \OptMSpanOf{\SubJobSet}
    \ge
    \max_{i\in\SubJobSet}\responseLen_i,
    \qquad
    \OptMSpanOf{\SubJobSet}
    \ge
    \frac{1}{\kvmem}
    \sum_{i\in\SubJobSet}\AreaOf{i}.
\end{equation*}
\end{lemma}

The first bound is the processing requirement of the longest job.
The second follows because all memory-time areas of jobs must be supplied.
We use the branch sets and capped small instance from \Cref{sec:routing-meta-scheduler}.

\begin{proposition}[Branch makespan bounds]
\label{prop:alternative-makespan-branches}
Set $\PromptCutoff=\kvmem/4$ and $\responseScalar=2$.
Then
\begin{equation*}
    \MSpan{\ALGRect(\LargeJobs)},
    \quad
    \MSpan{\ALGRect(\CappedSmallJobs)},
    \quad
    \MSpan{\ALGResp(\TrueResponsePool)}
\end{equation*}
are all $O(1) \cdot \OptMSpanOf{\JobSet}$.
\end{proposition}

\begin{proof}
For either rectangle branch instance $\SubJobSet$, \Cref{lem:rect-per-job} implies, after taking the maximum over jobs in $\SubJobSet$,
\begin{equation*}
    \MSpan{\ALGRect(\SubJobSet)}
    =
    O\left(
        \max_{i\in\SubJobSet}\processingTime_i
        +
        \frac1\kvmem
        \sum_{i\in\SubJobSet}\RectAreaOf{i}
    \right),
\end{equation*}
using the same width-range conditions as in the proofs of \Cref{prop:rectangle-large-guarantee,prop:rectangle-prompt-heavy-guarantee}.
For the \LargeBranch, $\RectAreaOf{i}\le4\AreaOf{i}$.
For the \PromptBranch, the capped rectangle area is at most $2\AreaOf{i}$.
Thus both rectangle branch makespans are $O(1)\cdot \OptMSpanOf{\JobSet}$ by \Cref{lem:alternative:makespan-lower-bounds}.

For the \ResponseBranch, the phase-duration estimate used in the proof of \Cref{thm:response-heavy-ratio}, summed over geometric response caps, gives
\begin{equation*}
    \MSpan{\ALGResp(\TrueResponsePool)}
    =
    O\left(
        \max_{i\in\TrueResponsePool}\responseLen_i
        +
        \frac{1}{\kvmem}
        \sum_{i\in\TrueResponsePool}\responseLen_i^2
    \right).
\end{equation*}
Here the geometric caps bound the linear contribution by the largest response length, and the phase-efficiency term contributes the quadratic area term.
Since every job in $\TrueResponsePool$ is \ResponseHeavy, $\responseLen_i^2\le2\AreaOf{i}$.
Another application of \Cref{lem:alternative:makespan-lower-bounds} completes the proof.
\end{proof}

The following theorem gives the makespan extension.

\begin{theorem}
\label{thm:alternative-makespan}
For every feasible batch KV-cache scheduling instance with arbitrary prompt lengths and arbitrary response lengths, there is a fully non-clairvoyant scheduling algorithm whose makespan is within a constant factor of the optimal clairvoyant makespan:
\begin{equation*}
    \MSpan{\ALG(\JobSet)}
    =
    O(1)\cdot \OptMSpanOf{\JobSet}.
\end{equation*}
\end{theorem}

\begin{proof}
Set $\PromptCutoff=\kvmem/4$ and $\responseScalar=2$, and run the routing meta-scheduler $\ALGRoute$ in \Cref{alg:white-box-meta}.
Let
\[
    T=
    \max\left\{
        \max_{i\in\LargeJobs}\Clarge_i,\;
        \max_{i\in\SmallJobs}\Cprompt_i,\;
        \max_{i\in\TrueResponsePool}\Cresponse_i
    \right\}.
\]
By the subset-control discussion in \Cref{sec:routing-meta-scheduler}, every current branch call completes or certifies its remaining jobs within the corresponding reference bound.
Choose the first stage $\Index$ with $2^\Index\ge T$.
In this stage, the \PromptBranch call completes all remaining jobs in $\TruePromptPool$ and certifies all remaining jobs in $\TrueResponsePool$.
The certified jobs then move to $\ResponsePool$, and the following \ResponseBranch call completes them; the \LargeBranch call completes all large jobs.
Thus all jobs finish by the end of the stage.
The total time through that stage is $O(2^\Index)=O(T)$.
The same bounds used in \Cref{prop:alternative-makespan-branches} give $T=O(1)\cdot\OptMSpanOf{\JobSet}$, which proves the theorem.
\end{proof}

\subsection{Total Completion Time with Online Arrivals}
\label{sec:alternative-online-completion-time}

Each job $i$ now has an arrival time $\arrivalTime_i\in\naturals$.
Before time $\arrivalTime_i$, job $i$ is not visible to the scheduler.
At time $\arrivalTime_i$, its prompt length $\initialLen_i$ becomes known, while its response length $\responseLen_i$ is still unknown until completion.
Let $\OptArrivalOf{\JobSet}$ be the optimal total completion time of a clairvoyant schedule that knows all jobs and response lengths in advance, but may process job $i$ only after time $\arrivalTime_i$.

\begin{lemma}[Online-arrival lower bounds]
\label{lem:alternative:arrival-lower-bounds}
\begin{equation*}
    \OptArrivalOf{\JobSet}
    \ge
    \OptOf{\JobSet},
    \qquad
    \OptArrivalOf{\JobSet}
    \ge
    \sum_{i\in\JobSet}\left(\arrivalTime_i+\responseLen_i\right).
\end{equation*}
\end{lemma}

The first bound follows by removing arrival constraints.
The second follows because job $i$ cannot complete before time $\arrivalTime_i+\responseLen_i$.

The online scheduler uses the same pool convention: newly arrived jobs enter either $\LargeJobs$ or $\PromptPool$, certified but unfinished jobs move to $\ResponsePool$, and completed jobs are removed from their branch pools.

\begin{algorithm2e}
\caption{Online-Arrival Routing Meta-Scheduler ($\ALGRouteOnline$)}
\label{alg:online-arrival-meta}
\SetKwInput{KwOnline}{Online Arrival}
\KwIn{Memory $\kvmem$ and prompt cutoff $\PromptCutoff$.}
\KwOnline{Job $i$ becomes visible at time $\arrivalTime_i$ and reveals its prompt length $\initialLen_i$.}
$\LargeJobs\leftarrow\varnothing$ \tcp*{\LargeBranch}
$\PromptPool\leftarrow\varnothing$ \tcp*{\PromptBranch}
$\ResponsePool\leftarrow\varnothing$ \tcp*{\ResponseBranch}
\For{stage $\Index=0,1,2,\ldots$}{
    \tcp{Admit new jobs.}
    \For{each newly arrived job that has not been assigned to a branch}{
        \eIf{$\initialLen_i>\PromptCutoff$}{
            Add $i$ to $\LargeJobs$\;
        }{
            Add $i$ to $\PromptPool$\;
        }
    }
    \tcp{Run current doubling stage.}
    Run a fresh copy of the \PromptBranch scheduler on $\PromptPool$ for $2^\Index$ rounds\;
    Move every newly certified job from $\PromptPool$ to $\ResponsePool$\;
    Run a fresh copy of the \ResponseBranch scheduler on $\ResponsePool$ for $2^\Index$ rounds\;
    Run a fresh copy of the \LargeBranch scheduler on $\LargeJobs$ for $2^\Index$ rounds\;
}
\end{algorithm2e}

The following theorem gives the online-arrival extension.

\begin{theorem}
\label{thm:alternative-online-completion}
For every feasible online-arrival KV-cache scheduling instance with arbitrary prompt lengths, arbitrary response lengths, and arbitrary arrival times, there is an online non-clairvoyant scheduling algorithm whose total completion time is within a constant factor of the optimal clairvoyant total completion time among schedules that process each job only after its arrival:
\begin{equation*}
    \ALGRouteOnline(\JobSet)
    =
    O(1)\cdot \OptArrivalOf{\JobSet}.
\end{equation*}
\end{theorem}

\begin{proof}
Set $\PromptCutoff=\kvmem/4$ and $\responseScalar=2$, and run the online-arrival routing meta-scheduler in \Cref{alg:online-arrival-meta}.
The online arrival constraint only delays the time at which a job can enter the routing process.
A job arriving at time $\arrivalTime_i$ is admitted at the start of the first stage that begins at or after $\arrivalTime_i$; by the doubling schedule, the elapsed time before that stage is $O(\arrivalTime_i)$.
After this prefix, the same routing calculation as in the batch proof applies to job $i$, using the reference bounds $\Clarge_i,\Cprompt_i,\Cresponse_i$ as in the proof of \Cref{thm:main}.
The current branch calls are on subsets of the corresponding reference instances, so the same subset arguments apply.
Thus
\begin{align*}
    \C_i^{\ALGRouteOnline}
    &=
    O(\arrivalTime_i+\Clarge_i)
    &\text{for every } i\in\LargeJobs,\\
    \C_i^{\ALGRouteOnline}
    &=
    O(\arrivalTime_i+\Cprompt_i)
    &\text{for every } i\in\TruePromptPool,\\
    \C_i^{\ALGRouteOnline}
    &=
    O(\arrivalTime_i+\Cprompt_i+\Cresponse_i)
    &\text{for every } i\in\TrueResponsePool.
\end{align*}
Summing over all jobs gives the same three branch sums as in the proof of \Cref{thm:main}, plus the arrival term $\sum_{i\in\JobSet}\arrivalTime_i$.
The branch sums are bounded by $O(1)\cdot \OptOf{\JobSet}$, while the arrival term is bounded by \Cref{lem:alternative:arrival-lower-bounds} and $\OptOf{\JobSet}\le\OptArrivalOf{\JobSet}$.
Therefore $\ALGRouteOnline(\JobSet)=O(1)\cdot \OptArrivalOf{\JobSet}$.
\end{proof}
 
\end{document}